\documentclass[10.6pt, reqno]{article}
\usepackage[utf8]{inputenc}
\usepackage{amsmath, amsthm}
\usepackage{amsfonts}
\usepackage{mathrsfs}
\usepackage{amssymb}
\usepackage{soul}
\usepackage{bbm}
\usepackage{todonotes}
\usepackage{soul}
\usepackage{authblk}
\usepackage{bm}
\usepackage{cite}
\usepackage{color}
\usepackage{float}
\usepackage{authblk}
\usepackage{authblk}
\usepackage{enumerate}

\usepackage{enumerate}
\usepackage[utf8]{inputenc}
\usepackage[normalem]{ulem}
\usepackage{amsmath, amsthm,mathrsfs}
\usepackage{amsfonts}
\usepackage{amssymb}
\usepackage{bbm}
\usepackage{cite}
\usepackage{enumerate}
\usepackage[utf8]{inputenc}
\usepackage[normalem]{ulem}
\usepackage{amsmath, amsthm,mathrsfs}
\usepackage{amsfonts}
\usepackage{amssymb}
\usepackage{bbm}
\usepackage{cite}

\usepackage{enumerate}
\usepackage{caption}
\usepackage{hyperref}
\usepackage{hyperref}
\usepackage{slashed}
\hypersetup{hypertex=true,
colorlinks = black,
linkcolor =black,
anchorcolor = black,
citecolor = black}
\usepackage{graphicx}
\usepackage{cancel}
\usepackage{mathtools}
\usepackage{subfigure}
\usepackage{slashed}
\usepackage{amsfonts}
\usepackage{amssymb}
\usepackage{accents}
\usepackage{amsmath}
\usepackage{amsxtra}
\usepackage{epstopdf}
\usepackage{latexsym}
\usepackage{mathrsfs}
\usepackage{leftidx}
\usepackage{tensor}
\usepackage{enumitem}
\newtheorem{theorem}{Theorem}[section]
\newtheorem{lemma}[theorem]{Lemma}
\newtheorem{corollary}[theorem]{Corollary}
\newtheorem{proposition}[theorem]{Proposition}
\newtheorem{conjecture}[theorem]{Conjecture}
\newtheorem{definition}[theorem]{Definition}

\theoremstyle{remark}
\newtheorem{remark}[theorem]{Remark}

\numberwithin{equation}{section}
\begin{document}
\title{Asymptotically Anti-de Sitter Spherically Symmetric Hairy Black Holes}

\author{Weihao Zheng\thanks{wz344@math.rutgers.edu}}
\affil{\small Department of Mathematics, Rutgers University, Hill Center, 110 Frelinghuysen Road, Piscataway, NJ, USA}
\date{\today}
\maketitle
\begin{abstract}
We construct one-parameter families of static spherically symmetric asymptotically anti-de Sitter black hole solutions $(\mathcal{M},g_{\epsilon},\phi_{\epsilon})$ to the Einstein--Maxwell--(charged) Klein--Gordon equations. Each family bifurcates off a sub-extremal Reissner--Nordström-AdS spacetime $(\mathcal{M},g_{0},\phi_{0}\equiv0)$. For a co-dimensional one set of black hole parameters, we show that Dirichlet (respectively Neumann) boundary conditions can be imposed for the scalar field. The construction provides a counter-example to a version of the no-hair conjecture in the context of a negative cosmological constant. Our result is based on our companion work [W. Zheng, \emph{Exponentially-growing Mode Instability on the Reissner--Nordström-Anti-de-Sitter black holes}], in which the existence of linear hair and growing mode solutions has been established. In the charged scalar field case, our result provides the first rigorous mathematical construction of the so-called holographic superconductors, which are of particular significance in high-energy physics. 
\end{abstract}
\section{Introduction}
\label{sec:intro}
The no-hair conjecture in General Relativity asserts that all asymptotically flat stationary black hole solutions to the Einstein equations in (electro-)vacuum can be completely characterized by their mass $M$, charge $Q$, and angular momentum $a$, namely, they belong to the Kerr--Newman family (Schwarzschild when $a = Q = 0$, Kerr when $Q = 0 $, and Reissner--Nordström when $a = 0$; see for example \cite{hawking2023large}). On the one hand, in the past decades, many uniqueness results have been established, providing a form of resolution to the no-hair conjecture in the asymptotically flat (electro-)vacuum; see \cite{carter1971axisymmetric,robinson1975uniqueness,hawking1972black,alexakis2010uniqueness,wong2014non}. On the other hand, it is reasonable to extend the no-hair conjecture to the Einstein equations coupled with simple matter models. For instance, for the Einstein--Maxwell--(charged) scalar field equations or Einstein--Maxwell--(charged) Klein--Gordon equations, within the spherical symmetry, it has been shown that the analogous uniqueness results still hold\cite{jd1972nonexistence}.\footnote{Nonetheless, for matter models which admit hair on flat spacetime or which exhibit superradiance, stationary black hole solutions violating at least the spirit of the no-hair conjecture has been constructed; for instance, see\cite{chodosh2017time,smoller1993existence}.}

In the present study, we are interested in the Einstein--Maxwell--(charged) Klein--Gordon equations with a negative cosmological constant $\Lambda$:
\begin{align}
&R_{\mu\nu}(g) - \frac{1}{2}g_{\mu\nu} + \Lambda g_{\mu\nu} = T_{\mu\nu}^{EM} + T^{KG}_{\mu\nu},\label{original 1}\\
&T_{\mu\nu}^{EM} = 2\left(g^{\alpha\beta}F_{\alpha\nu}F_{\beta\mu} - \frac{1}{4}F^{\alpha\beta}F_{\alpha\beta}g_{\mu\nu}\right),\quad \nabla^{\mu}F_{\mu\nu} = iq_{0}\left(\frac{\phi\overline{D_{\nu}\phi}-\bar{\phi}D_{\nu}\phi}{2}\right),\ F = dA\label{original 2}\\
&T_{\mu\nu}^{KG} = 2\left(D_{\mu}\phi \overline{D_{\nu}\phi} - \frac{1}{2}\left(g^{\alpha\beta}D_{\alpha}\phi\overline{D_{\beta}\phi} + \bigl(-\frac{\Lambda}{3}\bigr)\alpha\vert\phi\vert^{2}\right)g_{\mu\nu}\right),\quad D_{\mu} = \nabla_{\mu}+iq_{0}A_{\mu},\label{original 3}\\
&g^{\mu\nu}D_{\mu}D_{\nu}\phi = \bigl(-\frac{\Lambda}{3}\bigr)\alpha\phi,\label{original 4}
\end{align}
where $\alpha$ is the negative Klein--Gordon mass, $q_{0}\in\mathbb{R}$ is the scalar field charge, and $1$-form $A$ is the electromagnetic potential. Reissner--Nordström-AdS $(\mathcal{M},g_{RN},\phi\equiv0)$ is known as the special solutions to \eqref{original 1}-\eqref{original 4}:
\begin{equation}
\label{brief intro of RN}
\begin{aligned}
&g_{RN} = -\Omega_{RN}^{2}dt^{2}+\frac{1}{\Omega_{RN}^{2}}dr^{2}+r^{2}d\sigma^{2},\\&
\Omega_{RN}^{2} = 1-\frac{2M}{r}+\frac{e^{2}}{r^{2}}+\bigl(-\frac{\Lambda}{3}\bigr)r^{2}.
\end{aligned}
\end{equation}
$M>0$ is the black hole mass and $e\in\mathbb{R}$ is the black hole charge. We say $g_{RN}$ is sub-extremal if $\Omega_{RN}$ admits two different positive roots. Let $r_{+}(M,e,\Lambda)$ be the largest positive root of $\Omega_{RN}^{2}$. Instead of using $(M,e,\Lambda)$ as parameters for $\Omega_{RN}^{2}$, this paper adopts $(M,r_{+},\Lambda)$ to parametrize Reissner--Nordström-AdS. In this framework, for fixed $(r_{+},\Lambda)$, the admissible sub-extremal range of $M$ is given by $M_{e = 0}\leq M<M_{0}$, where $M_{e = 0}$ corresponds to $e = 0$ and $M_{0}$ corresponds to the extremality. See Section \ref{pre of Reissner} for a detailed discussion.

In the context of a negative cosmological constant, due to the superradiance on spacetimes that are non-spherically symmetric or with a charged scalar field ($q_{0}\neq0$ in \eqref{original 1}-\eqref{original 4}), mathematical and numerical works \cite{dold2017unstable,hartnoll2008holographic,horowitz2011introduction,hartnoll2008building,hartnoll2020gravitational} suggest that no-hair conjecture might be false; we will come back to this later in Section \ref{sec:superradiance instability}. However, it is reasonable to formulate the following no-hair conjecture for the spherically symmetric system with a neutral scalar field:

\begin{conjecture}
\label{no hair AdS}
The only stationary spherically symmetric solutions solving $\eqref{original 1}$-\eqref{original 4} with $q_{0} = 0$ under Dirichlet (respectively Neumann) boundary conditions\footnote{Due to the lack of global hyperbolicity of the asymptotically anti-de Sitter spacetime, the natural formulation is the initial-boundary value problem. see Section \ref{boundary} for the definition of various boundary conditions.} belong to the Reissner--Nordström-AdS family \eqref{brief intro of RN}.
\end{conjecture}
Notably, Conjecture \ref{no hair AdS} is proved to be true for the spherically symmetric Einstein--Klein--Gordon equations (\eqref{original 1}-\eqref{original 4} with $F\equiv 0 $ and $q_{0} = 0$) in the work of Holzegel and Smulevici \cite{holzegel2013stability}. However, in this paper, without assuming $F = 0$, we provide a counter-example to Conjecture \ref{no hair AdS}, as a part of more general construction that we now state:


\begin{theorem}\textup{[Rough Version of the Main Theorem]}
\label{rough version 3}
Imposing Dirichlet (respectively Neumann) boundary conditions for the scalar field $\phi$ and letting $-\frac{9}{4}<\alpha<0$ ($-\frac{9}{4}<\alpha<-\frac{5}{4}$ for Neumann boundary conditions respectively), then we have the following two results about black hole solutions to \eqref{original 1}-\eqref{original 4} bifurcating off a Reissner--Nordström-AdS spacetime:
\begin{enumerate}\label{condition}
\item[(a)]\textup{(Large charge case)} Assume $(M_{b},r_{+},\Lambda)$ are given sub-extremal parameters and $\alpha$ is a fixed Klein--Gordon mass. Then there exists $q^{max} = q^{\max}(M_{b},r_{+},\Lambda,\alpha)$, such that for any $\vert q_{0}\vert>q^{\max}$, there exists $M_{e = 0}<M_{c} = M_{c}(M_{b},r_{+},\Lambda,\alpha,q_{0})<M_{b}$, such that there exist stationary black hole solutions to \eqref{original 1}-\eqref{original 4} with $\phi\neq0$ (so-called hairy), bifurcating off the Reissner--Nordström-AdS with the sub-extremal parameters $(M_{c},r_{+},\Lambda)$.
\item[(b)]\textup{(Fixed charge condition)} Let $R_{0}$ be the positive solution to the quadratic equation:\begin{equation}
24\left(\alpha+\frac{3}{2}-\frac{q_{0}^{2}}{2\bigl(-\frac{\Lambda}{3}\bigr)}\right)x^{2}+\left(4\left(\alpha+\frac{3}{2}-\frac{q_{0}^{2}}{2\bigl(-\frac{\Lambda}{3}\bigr)}\right)-\frac{2q_{0}^{2}}{\bigl(-\frac{\Lambda}{3}\bigr)}+6\right)x+1 = 0.\label{quadratic}
\end{equation}
For any fixed parameters $(r_{+},\Lambda,\alpha,q_{0})$ satisfying \begin{align}
&-\frac{9}{4}<\alpha<-\frac{3}{2}+\frac{q_{0}^{2}}{2\bigl(-\frac{\Lambda}{3}\bigr)},\label{large1}\\
&\bigl(-\frac{\Lambda}{3}\bigr)r_{+}^{2}>R_{0}\label{large2},
\end{align}
there exists $M_{e = 0}<M_{c}(r_{+},\Lambda,\alpha,q_{0})<M_{0}$, such that there exist black hole solutions to \eqref{original 1}-\eqref{original 4} with $\phi\neq0$, bifurcating off the Reissner--Nordström-AdS with the sub-extremal parameters $(M_{c},r_{+},\Lambda)$.
\end{enumerate}
\end{theorem}



\begin{remark}
One could also consider the Robin boundary condition, which is a given linear combination of Dirichlet and Neumann; see Section \ref{boundary} for the precise definition there. However, we will not pursue the problem under the Robin boundary condition in this paper.
\end{remark}

\paragraph{Decay of the scalar field and the asymptotic stability}
In the work of Holzegel and Smulevici \cite{holzegel2013stability}, they proved that the solutions $(\mathcal{M},g,\phi)$ to the spherically symmetric Einstein--Klein--Gordon equations with near-Schwarzschild-AdS initial data and under Dirichlet boundary conditions will converge to Schwarzschild-AdS exponentially. Their result shows the asymptotic stability of Schwarzschild-AdS as solutions to the spherically symmetric Einstein--Klein--Gordon equations. In contrast to their result, one immediate corollary of our Theorem \ref{rough version 3} is the following:

\begin{corollary}
Reissner--Nordström-AdS spacetimes are not asymptotically stable as solutions to spherically symmetric equations \eqref{original 1}-\eqref{original 4}. In particular, there exist sub-extremal Reissner--Nordström-AdS spacetimes and arbitrary small spherical perturbations to the initial data of \eqref{original 1}-\eqref{original 4} such that the scalar field does not decay to $0$. 
\end{corollary}
\begin{remark}
In view of the fact that Reissner--Nordström-AdS is a co-dimensional three subset of the Kerr--Newman-AdS family, one could expect finite co-dimensional stability of Reissner--Nordström-AdS as solutions to \eqref{original 1}-\eqref{original 4}, analogously to the co-dimensional three stability of Schwarzschild in \cite{dafermos2021non}. However, Theorem \ref{rough version 3} suggests that Reissner--Nordström-AdS is likely not co-dimensional three stable as solutions to \eqref{original 1}-\eqref{original 4}. Similarly, despite the stability of Schwarzschild in $U(1)$ symmetry \cite{klainerman2017global}, it is unlikely that Reissner--Nordström-AdS is asymptotically stable among $U(1)$ perturbations.
\end{remark}
One can refer to Section \ref{sec:stability} for a more detailed introduction regarding the connection between the hairy black hole solutions and the stability problem. 

\paragraph{The Klein--Gordon equation on the Reissner--Nordström-AdS spacetime}

The Klein--Gordon equation on the Reissner--Nordström-AdS spacetime takes the form of\begin{equation}
D_{\mu}D^{\mu}\phi = \bigl(-\frac{\Lambda}{3}\bigr)\alpha\phi,\quad D_{\mu} = \nabla_{\mu}+iq_{0}A_{\mu},\label{equ:kl}
\end{equation}
where $A = e\left(\frac{1}{r_{+}}-\frac{1}{r}\right)$. The existence of time-independent solutions and growing mode solutions to \eqref{equ:kl} under Dirichlet (respectively Neumann) boundary conditions has been established in our companion work \cite{weihaozheng}. The construction of hairy black hole solutions in this paper under Dirichlet (respectively Neumann) boundary conditions is based upon the results and methods in \cite{weihaozheng}. Although the construction of the stationary solution in \cite{weihaozheng} is implicit, we can still obtain the following necessary condition on the black hole charge for a stationary solution to the Klein--Gordon equation to exist:
\begin{equation}
\left(1+\frac{4q_{0}^{2}}{\bigl(-\frac{\Lambda}{3}\bigr)}\right)\frac{e^{2}}{r_{+}^{2}}>1+R_{0}.
\label{large charge condition}
\end{equation}

\paragraph{The Breitenlohner--Freedman bound and interpretation of conditions \eqref{large1} and \eqref{large2}}
The Breitenlohner--Freedman bound $\alpha>-\frac{9}{4}$ plays an important role in the local well-posedness and stability issues for the equations in the asymptotically AdS spacetime; see the original work of Breitenlohner and Freedman \cite{breitenlohner1982stability}; see also the related mathematics work \cite{holzegel2012self,holzegel2013stability,holzegel2013decay,holzegel2014boundedness,warnick2013massive,holzegel2012well}.

Theorem \ref{rough version 3} is preceded by physics heuristics, predicting the existence of hairy black holes bifurcating the extremal Reissner--Nordström-AdS spacetime \cite{hartnoll2008holographic,gubser2008breaking,hartnoll2008building,horowitz2011introduction}. The condition \eqref{large1} in the above theorem has been confirmed earlier by physicists \cite{hartnoll2008holographic} heuristically, using the near-horizon geometry of the extremal asymptotically AdS spacetime. Note the condition \eqref{large charge condition} excludes the existence of hairy black hole solutions for small black hole charge $e$, which is consistent with the uncharged spherically symmetric stability result for Schwarzschild-AdS in \cite{holzegel2013stability}.

\subsection{Relation between the hairy black hole solutions and the stability problem}
\label{sec:stability}
In this section, we will focus on discussing the relation between hairy black hole solutions and stability problems.
\subsubsection{Stability results on asymptotically flat spacetimes and co-dimensional three stability of Schwarzschild}
\label{sec:codimension three stability}
We start this section by discussing some stability results on asymptotically flat spacetimes. Recall that under the Boyer--Lindquist coordinates, the Kerr--Newman-(dS/AdS) metric solving \eqref{original 1}-\eqref{original 4} takes the form of:
\begin{equation}
\label{Kerr-Newman}
\begin{aligned}
g_{a,e,M}^{(dS/AdS)} =& -\frac{\Delta-\Delta_{\theta}a^{2}\sin^{2}\theta}{\Sigma}dt^{2}-2\frac{(r^{2}+a^{2})\Delta_{\theta}-\Delta}{\Xi\Sigma}a\sin^{2}\theta dtd\varphi+\frac{\Sigma}{\Delta}dr^{2}\\&+\frac{\Sigma}{\Delta_{\theta}}d\theta^{2}+\frac{\Delta_{\theta}(r^{2}+a^{2})^{2}-\Delta a^{2}\sin^{2}\theta}{\Xi^{2}\Sigma}d\varphi^{2},\\
\Delta =& (r^{2}+a^{2})\left(\bigl(-\frac{\Lambda}{3}\bigr)r^{2}+a^{2}\right)-2Mr+e^{2},\quad \Sigma = r^{2}+a^{2}\cos^{2}\theta,\\\quad \Xi =& 1-\bigl(-\frac{\Lambda}{3}\bigr)a^{2},\quad \Delta_{\theta} = 1+\bigl(-\frac{\Lambda}{3}\bigr)a^{2}\cos^{2}\theta.
\end{aligned}
\end{equation}
\begin{itemize}
  \item The case $\Lambda = 0$ in \eqref{Kerr-Newman} corresponds to the asymptotically flat Kerr--Newman family \cite{hawking2023large}; the case $\Lambda>0$ corresponds to the Kerr--Newman-dS family; the case $\Lambda<0$ corresponds to the Kerr--Newman-AdS family.

  \item When $a = 0$ and $\Lambda<0$, \eqref{Kerr-Newman} will be reduced to the Reissner--Nordström-AdS spacetimes \eqref{brief intro of RN}; when $e=0$ and $\Lambda<0$, \eqref{Kerr-Newman} will be reduced to the Kerr-AdS spacetimes; when $e = a = 0$ and $\Lambda<0$, \eqref{Kerr-Newman} will be reduced to the Schwarzschild-AdS spacetimes.
  \item Since Reissner--Nordström-AdS is a co-dimensional three sub-family of Kerr--Newman-AdS, one may expect the co-dimensional three stability of Reissner--Nordström-AdS as solutions to \eqref{original 1}-\eqref{original 4}, analogously to \cite{dafermos2021non}; or stability under $U(1)$ perturbations, analogously to \cite{klainerman2017global}, in the context of $\Lambda= 0$ and the Einstein vacuum equations. However, the existence of hairy black hole solutions to \eqref{original 1}-\eqref{original 4} bifurcating off Reissner--Nordström-AdS strongly suggests that the moduli space of initial data for \eqref{original 1}-\eqref{original 4} leading to Reissner--Nordström-AdS has co-dimension strictly larger than three, in view of the additional presence of scalar hair $\phi$.
  \item It is also reasonable to make the conjecture that Reissner--Nordström-AdS is neither co-dimensional three stable or stable under $U(1)$ perturbations as solutions to Einstein--Maxwell equations (namely \eqref{original 1}-\eqref{original 4} without a scalar field) since the equations of linearized gravity for $\Lambda<0$ resemble Klein--Gordon equation on Reissner--Nordström-AdS.
  \item Our construction of hairy black holes bifurcating off Reissner--Nordström-AdS crucially relies on the presence of a Maxwell field in view of \eqref{large charge condition}, therefore, does not impact on any hypothetical finite co-dimensional stability statement of Schwarzschild-AdS. On that note, we remark that the linear decay of massive waves on Schwarzschild-AdS \cite{holzegel2013decay} and nonlinear decay of solutions to spherically symmetric equations \eqref{original 1}-\eqref{original 4} with $F\equiv0$ near Schwarzschild \cite{holzegel2013stability}, as discussed in Section \ref{sec:asymptotic stability}, are consistent with the co-dimensional three stability within the moduli space of solutions to the Einstein--Klein--Gordon equations. Nevertheless, due to the stable trapping caused by the Dirichlet boundary condition and the weak turbulence phenomenon, both co-dimensional three stability statements of Schwarzschild-AdS as solutions to the Einstein--Klein--Gordon and the Einstein vacuum equations are likely false; see Section \ref{sec:superradiance instability} for further discussion.
\end{itemize}
\subsubsection{Asymptotic stability of Schwarzschild-AdS as solutions to the spherically symmetric Einstein--Klein--Gordon equations and decay of the field}
\label{sec:asymptotic stability}
In this section, we discuss about the stability results for the spherically symmetric Einstein--Klein--Gordon equations (\eqref{original 1}-\eqref{original 4} with $F\equiv 0$) under Dirichlet boundary condition established in the works of Holzegel and Smulevici \cite{holzegel2012self,holzegel2013stability},
\begin{itemize}
  \item In \cite{holzegel2012self}, Holzegel and Smulevici proved the local well-posedness result for the spherically symmetric Einstein--Klein--Gordon equations when the Klein--Gordon mass satisfies the Breitenlohner--Freedman bound $\alpha>-\frac{9}{4}$. 
  \item In \cite{holzegel2013stability} they showed that the Schwarzschild-AdS spacetimes are stable solutions to the spherically symmetric Einstein--Klein--Gordon equations under Dirichlet boundary conditions for $\alpha\geq-2$. Under a small spherically symmetric perturbation to the Schwarzschild-AdS initial data, the maximal development of the spacetime stays close to the Schwarzschild-AdS and converges to the Schwarzschild-AdS exponentially. In particular, the scalar field $\phi$ decays exponentially in time. 
\item Although in \cite{holzegel2013stability}, Holzegel and Smulevici only considered the case of zero Maxwell field $F\equiv0$, their work can be easily generalized to \eqref{original 1}-\eqref{original 4} with a small Maxwell field. However, in this paper, we consider the more general case, where the Maxwell field $F$ is not necessary zero or small. In Theorem \ref{rough version 3}, we show that Holzegel and Smulevici's stability result \cite{holzegel2013stability} does not hold in the presence of a large Maxwell field. In fact, \eqref{large charge condition} can be viewed as quantification of the necessary largeness of the Maxwell field.
\item In contrast to the decay result of the scalar field in \cite{holzegel2013stability}, our result shows that for solutions to \eqref{original 1}-\eqref{original 4} with Dirichlet boundary condition and near-Reissner--Nordström-AdS initial data, there is no decay for the scalar field. 
\end{itemize}
\subsubsection{Linearized stability of Schwarzschild-AdS as solutions to the Einstein vacuum equations}
In a series of works \cite{graf2024linear1,graf2024linear}, Holzegel and Graf proved the linear stability of the Schwarzschild-AdS spacetimes under Dirichlet-type boundary conditions. Under a small perturbation to Schwarzschild-AdS initial data, they proved that the solutions to the Teukolsky equations converge to Kerr-AdS inverse logarithmically in time. However, this decay rate is believed to be too slow to have the nonlinear stability of Schwarzschild-AdS as solutions to the Einstein vacuum equations; see the detailed discussion of this aspect in Section \ref{sec:superradiance instability}.
\subsubsection{Superradiant instability of Klein--Gordon equations on the Kerr and Kerr-AdS spacetimes and other instability mechanisms}
\label{sec:superradiance instability}
In contrast to spherically symmetric spacetimes or charged spacetimes with a charged scalar field, the well-known rotation-induced or charge-induced superradiance phenomenon may\footnote{Nonetheless, despite the presence of superradiance phenomenon for Kerr spacetime, decay results for the massless wave equations on Kerr have been established \cite{shlapentokh2015quantitative,dafermos2016decay,shlapentokh2020boundedness,shlapentokh2023boundedness,tataru2013local}.
} 
generate instability. Due to the superradiance and the presence of the Klein--Gordon mass, Shlapentokh-Rothman
\cite{shlapentokh2014exponentially} proved the existence of growing mode and oscillating mode solutions to the Klein--Gordon equation on Kerr (\eqref{Kerr-Newman} with $e = \Lambda = 0$). Later Shlapentokh-Rothman and Chodosh \cite{chodosh2017time} further constructed families of stationary black hole metric with a time-periodic scalar field solving Einstein--Klein--Gordon equations and bifurcating off the Kerr spacetimes. In direct contrast to the uniqueness results of Kerr and Kerr-Newman established in vacuum \cite{carter1971axisymmetric,robinson1975uniqueness,hawking1972black,alexakis2010uniqueness,wong2014non}, their result provides a counter-example to the analogous no-hair conjecture for Kerr as solutions to the Einstein--Klein--Gordon equations. Moreover, their construction shows the failure of asymptotical stability of Kerr as solutions to the Einstein--Klein--Gordon equations.

Similarly, the superradiance phenomenon appears on Kerr-AdS when the so-called Hawking--Reall bound $r_{+}^{2}>\vert a\vert\bigl(-\frac{\Lambda}{3}\bigr)^{-\frac{1}{2}}$ is violated, where $r_{+}$ here is the area radius of the event horizon of Kerr-AdS. Indeed, Dold \cite{dold2017unstable} proved the existence of growing mode solutions to the Klein--Gordon equation on Kerr-AdS violating the Hawking--Reall bound under Dirichlet and Neumann boundary conditions respectively. Since the Klein--Gordon equation with a negative mass on Kerr-AdS is believed to be a toy model of Einstein vacuum equations, Dold's result suggests that Kerr-AdS spacetimes violating the Hawking--Reall bound are unstable solutions to the Einstein vacuum equations. 

For the charge-induced superradiance, growing mode solutions have also been constructed for the Klein--Gordon equation on Kerr--Newman-dS and Reissner--Nordström-AdS \cite{besset2021existence}.

It should be noted that there is no analogous superradiance on Reissner--Nordström-AdS with a uncharged scalar field. However, condition \eqref{large2} can be viewed as an analogue of the Hawking--Reall bound. Its physical interpretation is up to debate.

Moreover, due to the stable trapping caused by the Dirichlet boundary condition \cite{holzegel2013decay,graf2023mode}, even within the Hawking--Reall bound, the sharp inverse logarithmical decay of solutions to the Klein--Gordon equations on Kerr-AdS \cite{holzegel2014quasimodes,holzegel2013decay} suggests the following non-linear instability conjecture of Kerr-AdS (see also \cite{figueras2023non}):
\begin{conjecture}\textup{\cite{holzegel2013decay}}
Kerr-AdS as solutions to the Einstein vacuum equations under Dirichlet boundary conditions are unstable against the generic perturbation of initial data.
\end{conjecture}

It is also interesting to mention that a plausible underlying instability mechanism, provided by the weak turbulence, has been suggested for a nonlinear scalar toy model on Schwarzschild-AdS \cite{KM}. We also remark that the weak turbulent instability of pure AdS spacetimes has already been discussed in \cite{bizon2011weakly}; see also the rigorous mathematics works regarding this aspect \cite{moschidis2020proof,moschidis2023proof}.

We emphasize here that the ``hairy'' mechanism for $q_{0} = 0$ in Theorem \ref{rough version 3} is new and unrelated to the superradiance, stable trapping, or weak turbulence. However, the case of large $\vert q_{0}\vert$ in Theorem \ref{rough version 3} is reminiscent of the charged superradiance; see \cite{di2015superradiance} for the discussion of $\Lambda = 0$.

\subsection{Relation to the black hole interior with the charged or uncharged scalar field}
We discuss the connection between our results about the hairy black hole exterior and previous results about the hairy black hole interior in this section. It is well-known that the maximal globally hyperbolic development of the characteristic Reissner--Nordström-AdS initial data on the event horizon will terminate at a null hypersurface called Cauchy horizon. However, Van de Moortel \cite{van2021violent}, and later Van de Moortel and Li \cite{li2023kasner} showed that for the interior of the hairy black hole bifurcating off Reissner--Nordström-AdS\footnote{In fact, their results apply to any value of $\Lambda$. We only discuss the case $\Lambda<0$ here.}, the spacetime solution to the Einstein--Maxwell--(charged) Klein--Gordon equations terminates at spacelike singularity described by a Kasner-like metric.
\begin{theorem}\textup{(Rough version of results in \cite{van2021violent,li2023kasner})}
\label{thm:interior}
Fix the following characteristic initial data on the bifurcating two-ended event horizon $\mathcal{H}_{1}^{+}\cup\mathcal{H}_{2}^{+}$:\begin{align*}
&\phi = \epsilon,\\&
g = g_{RN}+O(\epsilon^{2}),
\end{align*}
where $g_{RN}$ is the Reissner-Nordström-AdS metric \eqref{brief intro of RN} with sub-extremal parameters $(M,e,\Lambda)$. Then for every $(M,e,\Lambda,\alpha<0)$, there exists $\epsilon>0$ which can be chosen to be arbitrarily small, the maximal globally hyperbolic future development of spatially homogeneous spherically symmetric spacetime $(\mathcal{M}_{\epsilon},g_{\epsilon})$ ends at a spacelike singularity $\mathcal{S} = \{r = 0\}$. Moreover, we have the following descriptions of the spacelike singularity:\begin{enumerate}
\item[(1)]\textup{(violent collapse \cite{van2021violent})}For \eqref{original 1}-\eqref{original 4} with an uncharged scalar field ($q_{0} = 0$), the spacelike singularity $\{r = 0\}$ can be approximated by a Kasner metric\begin{equation}
\label{Kasner}
g_{K} = -d\tau^{2}+\tau^{2p_{1}(\epsilon)}dx_{1}^{2}+\tau^{2p_{2}(\epsilon)}dx_{2}^{2}+\tau^{2p_{3}(\epsilon)}dx_{3}^{2},\quad p_{1}+p_{2}+p_{3} = 1
\end{equation}
 of positive exponents $(p_{1},p_{2},p_{3}) = (1-C\epsilon^{2},C\epsilon^{2},C\epsilon^{2})+O(\epsilon^{3})$ for a constant $C>0$.
\item[(2)]\textup{(fluctuating collapse and Kasner inversion \cite{li2023kasner})}For \eqref{original 1}-\eqref{original 4} with a charged scalar field ($q_{0}\neq0$), the spacelike singularity $\{r = 0\}$ can be approximated by a Kasner metric \eqref{Kasner} of positive exponents $(1-2p(\epsilon),p(\epsilon),p(\epsilon))$, where $p(\epsilon)$ is a function of a highly oscillating quantity $\alpha(\epsilon)$ of the form\begin{equation*}
\alpha(\epsilon)\approx C\sin(\omega_{0}\epsilon^{-2}+O(\log(\epsilon^{-1})))\ \textup{as } \epsilon\rightarrow0
\end{equation*}
for $C\neq0$ and $\omega_{0}\in\mathbb{R}-\{0\}$
\end{enumerate}
\end{theorem}
\begin{itemize}
  \item In fact, results in \cite{van2021violent,li2023kasner} apply for any value of $\Lambda$ and Klein--Gordon mass $\alpha$. However, in \cite{van2021violent}, for general values of $(M,e,\Lambda,\alpha)$ with $(M,e,\Lambda)$ sub-extremal, the result only holds for almost every $(M,e,\Lambda,\alpha)$ with an exceptional set of measure zero. This exceptional set can be proved to be empty if $\alpha<0$ \cite{weihaozheng}. Since we are interested in the $\Lambda<0$ and $\alpha<0$ in this paper, the result in \cite{van2021violent} actually holds for every $(M,e,\Lambda,\alpha)$ when $\alpha<0$.
  \item Theorem \ref{thm:interior} mentioned in this section is the rough version of results in \cite{van2021violent,li2023kasner}. In fact in \cite{van2021violent}, Van de Moortel shows that for all values of $\epsilon$ which is small enough, the spacetime terminates at a spacelike singularity such that description in Theorem \ref{thm:interior} holds. However in \cite{li2023kasner}, Van de Moortel and Li have a detailed descrition of the set of $\epsilon$ for which of Theorem \ref{thm:interior} is true. Moreover, for some value of $\epsilon$, the so-called Kasner inversion appears.

  \item However, the existence of two-ended asymptotically AdS black holes satisfying the assumption of Theorem \ref{thm:interior} was unknown; see Open Problem $\romannumeral4$ in \cite{van2021violent}. Our main Theorem $\ref{rough version 3}$ can be viewed as the desired construction of the hairy black hole exterior satisfying these conditions, thus combining it with Van de Moortel and Li's results gives full pictures of the charged and uncharged hairy black hole exterior and interior.
\end{itemize}

\subsection{AdS/CFT correspondence and other asymptotically AdS hairy black solutions}
The (in)stability of the asymptotically AdS spacetime has gained growing attention due to the development of the string theory. The celebrated AdS/CFT correspondence has been proposed and developed since 1997 \cite{maldacena1999large,maldacena2003eternal}, which establishes a connection between a ($d$ dimensional) conformal field and a ($d+1$ dimensional) anti-de Sitter spacetime. One of the important applications of the AdS/CFT correspondence is that the hairy black hole solutions to the Einstein--Maxwell--Klein--Gordon equations can serve as a model of the holographic superconductor, the existence of which has been discussed heuristically in the physics literature \cite{horowitz2011introduction,gubser2008breaking,hartnoll2008building}.

We mention briefly the boundary conditions here. We say that the scalar field $\phi$ satisfies the Dirichlet boundary condition if $\phi\approx r^{-\frac{3}{2}-\sqrt{\frac{9}{4}}+\alpha}$ when $r\rightarrow\infty$; $\phi$ satisfies the Neumann boundary condition if $\phi\approx r^{-\frac{3}{2}+\sqrt{\frac{9}{4}}+\alpha}$. In the model of holographic superconductors, the existence of hairy black hole solutions to \eqref{original 1}-\eqref{original 4} under the Dirichlet boundary condition is often referred to as spontaneous condensation in the physics literature \cite{hartnoll2020gravitational} while the Neumann boundary condition corresponds to a source at the AdS conformal boundary and the existence of hairy black holes is called the stimulated condensation. In \cite{hartnoll2008holographic}, the authors have discussed the existence of the uncharged hairy black solutions to \eqref{original 1}-\eqref{original 4} under Neumann boundary conditions. Given the supporting numerics \cite{hartnoll2008holographic}, the existence of hairy black holes with Dirichlet boundary conditions is also conjectured to exist. We emphasize that our Theorem \ref{rough version 3} can be viewed as the first rigorous construction of this ``hairy'' holographic superconductor under Dirichlet boundary conditions.

It is interesting to mention that some other asymptotically AdS black hole spacetimes violating the spirit of the no-hair conjecture may exist. For the Einstein vacuum equations, it is conjectured that there are asymptotically AdS time-periodic black hole solutions with the Hawking--Reall bound saturated called ``black resonators'' \cite{dias2015black,fodor2015self}. Moreover, the asymptotically AdS, spatially localized and time-periodic solutions without an event horizon called ``geons'' have also been introduced in \cite{horowitz2014geons,fodor2017anti}. We emphasize that it is an open problem to construct these non-trivial spacetimes rigorously.
\section{Acknowledgements}
The author would like to express his gratitude to his advisor Maxime Van de Moortel, for his kind support, continuous encouragement and numerous inspiring discussions.

\section{Preliminary}
\label{Pre}
\subsection{The static spherically symmetric Einstein--Maxwell--(charged)Klein--Gordon equations in double null coordinates}
In this section, we introduce the Einstein--Maxwell--(charged) Klein--Gordon equations with a negative cosmological constant $\Lambda$ and the corresponding PDE system one can derive by using the spherical symmetry. Recall that the Einstein--Maxwell--(charged) Klein--Gordon equations take the following form:\begin{align}
&R_{\mu\nu}(g) - \frac{1}{2}g_{\mu\nu} + \Lambda g_{\mu\nu} = T_{\mu\nu}^{EM} + T^{KG}_{\mu\nu},\\
&T_{\mu\nu}^{EM} = 2\left(g^{\alpha\beta}F_{\alpha\nu}F_{\beta\mu} - \frac{1}{4}F^{\alpha\beta}F_{\alpha\beta}g_{\mu\nu}\right),\quad \nabla^{\mu}F_{\mu\nu} = iq_{0}\left(\frac{\phi\overline{D_{\nu}\phi}-\bar{\phi}D_{\nu}\phi}{2}\right),\ F = dA\\
&T_{\mu\nu}^{KG} = 2\left(D_{\mu}\phi \overline{D_{\nu}\phi} - \frac{1}{2}\left(g^{\alpha\beta}D_{\alpha}\phi\overline{D_{\beta}\phi} + \bigl(-\frac{\Lambda}{3}\bigr)\alpha\vert\phi\vert^{2}\right)g_{\mu\nu}\right),\quad D_{\mu} = \nabla_{\mu}+iq_{0}A_{\mu},\\
&g^{\mu\nu}D_{\mu}D_{\nu}\phi = \bigl(-\frac{\Lambda}{3}\bigr)\alpha\phi.
\end{align}
where $g$ is the metric of the manifold $M$, and $\phi$ is the (charged) scalar field. In the well-known double null coordinates, we can write the spherically symmeric metric $g$ as:
\begin{equation*}
g = -\Omega^{2}(u,v)dudv + r(u,v)^{2}d\sigma^{2},
\end{equation*}
where $d\sigma^{2} = r^{2}d\theta^{2}+\sin^{2}\theta d\varphi^{2}$ is the standard sphere metric. Computing the geometric quantities, we can write the spherically symmetric Einstein--Maxwell--(charged) Klein--Gordon equations with a negative cosmological constant $\Lambda$ as:
\begin{align}
\left(\frac{r_{u}}{\Omega^{2}}\right)_{u} &= -\frac{r}{\Omega^{2}}\vert D_{u}\phi\vert^{2}\label{Ray equation1},\\
\left(\frac{r_{v}}{\Omega^{2}}\right)_{v} &= -\frac{r}{\Omega^{2}}\vert D_{v}\phi\vert^{2},\label{Ray equation2}\\
\partial_{u}\partial_{v}r &= -\frac{r_{u}r_{v}}{r} - \frac{\Omega^{2}}{4r} - \frac{(-\Lambda)r\Omega^{2}}{4} + \bigl(-\frac{\Lambda}{3}\bigr)\frac{\alpha r\Omega^{2}\vert\phi\vert^{2}}{4} + \frac{Q^{2}\Omega^{2}}{4r^{3}},\label{Wave1}\\
\partial_{u}\partial_{v}\log\Omega^{2} &= \frac{\Omega^{2}}{2r^{2}} + 2\frac{r_{u}r_{v}}{r^{2}} - \frac{Q^{2}\Omega^{2}}{r^{4}} - 2\mathfrak{Re}(D_{u}\phi \overline{D_{v}\phi}),\label{Wave2}\\
\partial_{u}Q& = -q_{0}r^{2}\mathfrak{Im}(\phi \overline{D_{u}\phi}),\\
\partial_{v}Q& = q_{0}r^{2}\mathfrak{Im}(\phi\overline{D_{v}\phi})\\
D_{u}D_{v}\phi &= -\frac{\partial_{u}r\cdot D_{v}\phi}{r}-\frac{\partial_{v}r\cdot D_{u}\phi}{r}+\frac{iq_{0}Q\Omega^{2}}{4r^{2}}\phi-\bigl(-\frac{\Lambda}{3}\bigr)\frac{\alpha\Omega^{2}\phi}{4},\\
\partial_{u}A_{v}-\partial_{v}A_{u}& = \frac{Q\Omega^{2}}{2r^{2}}\label{Aequ}
\end{align}

The equations $\eqref{Ray equation1}$ and $\eqref{Ray equation2}$ are called the Raychaudhuri equations. And the equations $\eqref{Wave1}$ and $\eqref{Wave2}$ are the wave equations for $r$ and $\Omega^{2}$. We define the Vaidya mass as\begin{equation}
-\frac{4r_{u}r_{v}}{\Omega^{2}} = 1-\frac{2\varpi}{r}+\frac{Q^{2}}{r^{2}}+\bigl(-\frac{\Lambda}{3}\bigr)r^{2}.
\end{equation}
Then the transport equations for the vaidya mass are:\begin{align}
\partial_{u}\varpi &= -2r^{2}\left(\frac{r_{v}}{\Omega^{2}}\right)\vert D_{u}\phi\vert^{2}+\frac{1}{2}\bigl(-\frac{\Lambda}{3}\bigr)\alpha r^{2}\vert\phi\vert^{2}r_{u}-q_{0}Qr\Im(\phi\overline{D_{u}\phi}),\\
\partial_{v}\varpi &= -2r^{2}\left(\frac{r_{u}}{\Omega^{2}}\right)\vert D_{v}\phi\vert^{2}+\frac{1}{2}\bigl(-\frac{\Lambda}{3}\bigr)\alpha r^{2}\vert\phi\vert^{2}r_{v}-q_{0}Qr\Im(\phi\overline{D_{v}\phi}).
\end{align}

\subsection{The Reissner--Nordström-AdS spacetime}
As already mentioned in Section \ref{sec:intro}, in the electro-vacuum ($\phi = 0$), the static spherically symmetric solutions to \eqref{Ray equation1}-\eqref{Aequ} are the Reissner--Nordström-AdS spacetime $g_{RN}$ with the Maxwell field $F = -\frac{e}{r^{2}}dt\wedge dr$. In this case, we have $Q\equiv e$ and $\varpi\equiv M$. Define\begin{equation*}
\frac{dr^{*}}{dr} = -\frac{1}{\Omega_{RN}^{2}},\quad \lim_{r\rightarrow\infty}r^{*}(r) = 0.
\end{equation*}
Then when $r\rightarrow r_{+}$, we have $r^{*}\rightarrow\infty$. Let $u = t-r^{*}$, $v = t+r^{*}$. Then the Reissner--Nordström-AdS metric in the double null coordinates $(u,v,\theta,\varphi)$ takes the form of\begin{equation*}
g_{RN} = -\frac{1}{2}\Omega_{RN}^{2}du\otimes dv-\frac{1}{2}\Omega_{RN}^{2}dv\otimes du+r^{2}d\sigma^{2},\quad u,v\in\bar{\mathbb{R}}.
\end{equation*}
When $r^{*}\rightarrow\infty$, we have the asymptotic\begin{equation*}
\Omega_{RN}^{2} \approx 2K_{+}e^{-2K_{+}r^{*}},
\end{equation*}
where $K_{+}$ is defined by\begin{equation}
2K_{+} = 2\left(\frac{\varpi(r_{+})}{r_{+}^{2}}-\frac{Q^{2}(r_{+})}{r_{+}^{3}}+\bigl(-\frac{\Lambda}{3}\bigr)r_{+}\right).\label{def of K}
\end{equation}
The so-called event horizon\begin{equation*}
\mathcal{H}_{L}^{+} := \{v = \infty\},\quad \mathcal{H}_{R}^{+} := \{u = -\infty\},\quad \mathcal{H}^{+} := \mathcal{H}_{L}^{+}\cup\mathcal{H}_{R}^{+}
\end{equation*}
can be attached to the spacetime. $K_{+}$ is called the surface gravity of the event horizon. To regularly extend the metric to the event horizon, let
\begin{equation}
U = e^{K_{+}u},\quad V = e^{-K_{+}v}.\label{def of UV}
\end{equation}
Then under $(U,V,\theta,\varphi)$ coordinates, $\mathcal{H}^{+}_{L}$ corresponds to $\{V = 0\}$ and $\mathcal{H}^{+}_{R}$ corresponds to $\{U = 0\}$. Then on the event horizon, we have\begin{equation*}
^{(U,v)}\Omega_{RN}^{2}(U = 0,v) = 2e^{-K_{+}v},\quad ^{(u,V)}\Omega_{RN}^{2}(u,V = 0) = -2e^{K_{+}u},\quad ^{(U,V)}{\Omega_{RN}^{2}}_{|\mathcal{H}^{+}} = \frac{2}{K_{+}}.
\end{equation*}

\subsection{Set-up of the initial data and the gauge choice}
In this section, we fix the gauge choice and pose the characteristic initial data on the event horizon $\mathcal{H}^{+}$. First, we fix the gauge choice of $\Omega^{2}$ as\begin{align}
^{(U,v)}\Omega^{2}(U = 0,v) &= 2e^{-2K_{+}v},\\
^{(u,V)}\Omega^{2}(u,V = 0) & = -2e^{2K_{+}u},
\end{align}
where $U,V$ are defined in \eqref{def of UV} and $K_{+}$ is defined in \eqref{def of K}. Then for $(U,V)$ coordinates, on the event horizon, we have \begin{equation*}
^{(U,V)}\Omega^{2}_{|\mathcal{H}^{+}} = -\frac{2}{K_{+}}.
\end{equation*}
Next we impose the following characteristic initial data on the event horizon
\begin{equation}
\label{initial condition}
\begin{aligned}
\varpi_{|\mathcal{H}^{+}} &= M,\\
Q_{|\mathcal{H}^{+}} & = e,\\
r_{|\mathcal{H}^{+}} & = r_{+},\\
\phi_{|\mathcal{H}^{+}}& = \epsilon.
\end{aligned}
\end{equation}

Then on the event horizon, the equation $\eqref{Wave1}$ becomes\begin{equation}
\begin{aligned}
\partial_{U}\partial_{V}r &= \frac{^{(U,V)}\Omega^{2}}{4}\left(-\frac{1}{r_{+}}+\Lambda r_{+}+\bigl(-\frac{\Lambda}{3}\bigr)\alpha r_{+}\epsilon^{2}+\frac{e^{2}}{r_{+}^{3}}\right)\\& = -\frac{1}{2K_{+}}\left(-2K_{+}+\bigl(-\frac{\Lambda}{3}\bigr)\alpha r_{+}\epsilon^{2}\right).
\end{aligned}
\end{equation}
We have\begin{equation}
\lim_{u\rightarrow-\infty}\frac{2r_{u}}{\Omega^{2}} =\lim_{v\rightarrow\infty}\frac{-2r_{v}}{\Omega^{2}} = 1-\frac{\bigl(-\frac{\Lambda}{3}\bigr)\alpha r_{+}}{2K_{+}}\epsilon^{2}. 
\end{equation}
Last, to fix the gauge freedom, we should also fix the gauge choice for $A$. Let\begin{equation}
UA_{U}+VA_{V} = 0.\label{gauge choice for A}
\end{equation}
\subsection{Symmetry of the spacetime and system of ODEs}
Let $u = t-s$, $v = t+s$. For the local well-posedness issue of \eqref{Ray equation1}-\eqref{Aequ} with respect to our gauge choice and characteristic initial data, one may appeal to the argument in Proposition $2.1$ in \cite{li2023kasner}. Moreover, we have the following lemma.
\begin{lemma}
Under the initial condition \eqref{initial condition} and gauge choice \eqref{gauge choice for A}, we have $\phi$ is everywhere real on the maximal development of the spacetime and \begin{equation}
\partial_{t}Q = \partial_{t}\phi = \partial_{t}r = \partial_{t}A = \partial_{t}\Omega^{2} = 0. \label{stationary}
\end{equation}
\end{lemma}
\begin{proof}
Note that the equations \eqref{Ray equation1}-\eqref{Aequ} are invariant under the coordinates transformation $U^{\prime} = \lambda U$, $V^{\prime} = \lambda^{-1} V$. Since our gauge choice and initial conditions are also invariant under this transformation, by the local well-posedness of the system, we have\begin{equation*}
\Gamma(U,V) = \Gamma(\lambda U,\lambda^{-1}V),\quad \Gamma\in\{r,\Omega,Q,A,\phi\}.
\end{equation*}
Hence we have \eqref{stationary}. Since $\partial_{t}Q = 0$, we have \begin{align*}
\partial_{t}Q &= \partial_{u}Q+\partial_{v}Q = q_{0}r^{2}\left(\Im(\phi \overline{D_{v}\phi})-\Im(\phi\overline{D_{u}\phi})\right)\\&=
q_{0}r^{2}\Im(\phi\partial_{s}\overline{\phi}) = 0.
\end{align*}
Letting $\phi = \Phi e^{i\theta}$, then we have $\partial_{s}\theta = 0$. Then we can make the gauge choice such that $\phi$ is everywhere real.
\end{proof}

Hence we can write \eqref{Ray equation1}-\eqref{Aequ} as a system of ODEs: 

\begin{align}
\frac{d}{ds}\left(-r\frac{dr}{ds}\right) &= -\Omega^{2}\left(1 - \frac{Q^{2}}{r^{2}}\right) + \Omega^{2}r^{2}\left(\Lambda + \bigl(-\frac{\Lambda}{3}\bigr)\alpha\phi^{2}\right),\label{emkg1}\\
\frac{d^{2}\log(\Omega^{2})}{ds^{2}} &= -2\left(\frac{d\phi}{ds}\right)^{2} - \frac{2\Omega^{2}}{r^{2}} + \frac{2}{r^{2}}\left(\frac{dr}{ds}\right)^{2} + \frac{4Q^{2}\Omega^{2}}{r^{4}}+2A^{2}q_{0}^{2}\phi^{2},\label{emkg2}\\
\frac{d\kappa^{-1}}{ds} &= \frac{r}{\Omega^{2}}\left(\left(\frac{d\phi}{ds}\right)^{2}+q_{0}^{2}A^{2}\phi^{2}\right),\label{emkg3}\\
\frac{d}{ds}\left(r^{2}\frac{d\phi}{ds}\right) &= \bigl(-\frac{\Lambda}{3}\bigr)\alpha r^{2}\Omega^{2}\phi-q_{0}^{2}A^{2}r^{2}\phi,\label{emkg4}\\
\frac{dQ}{ds}& = -q_{0}^{2}r^{2}A\phi^{2},\label{Q equ}\\
\frac{dA}{ds} &= -\frac{Q\Omega^{2}}{r^{2}}\label{A equ}
\end{align}
where $\kappa^{-1} = -\frac{1}{\Omega^{2}}\frac{dr}{ds}$. Note here the equation $\eqref{emkg3}$ corresponds to the Raychaudhuri equations $\eqref{Ray equation1}$ and $\eqref{Ray equation2}$. 

Then the Vaidya mass $\varpi$ can be written as:
\begin{equation}
\Omega^{2}\kappa^{-2} = 1 - \frac{2\varpi}{r} + \frac{Q^{2}}{r^{2}} + \left(-\frac{\Lambda}{3}\right)r^{2}.
\end{equation}

From the above system $\eqref{emkg1}-\eqref{emkg4}$, we can write the transport equation for the Vaidya mass as:

\begin{equation}\label{Vaidya}
\frac{d\varpi}{ds} = -\frac{r^{2}}{2}\kappa^{-1}\left(\left(\frac{d\phi}{ds}\right)^{2}+q_{0}^{2}A^{2}\phi^{2}\right) + \bigl(-\frac{\Lambda}{3}\bigr)\alpha\frac{r^{2}}{2}\phi^{2}\frac{dr}{ds}-q_{0}^{2}QrA\phi^{2}.
\end{equation}
Next we consider the compatible initial data on the event horizon $\{s = \infty\}$:\begin{align}
&\lim_{s\rightarrow\infty}r(s) = r_{+},\quad \lim_{s\rightarrow\infty}Q(s) = e,\quad \lim_{s\rightarrow\infty}\varpi(s) = M,\quad \lim_{s\rightarrow\infty}\phi(s) = \epsilon,\\
&\lim_{s\rightarrow\infty}\kappa^{-1}(s) = 1-\frac{\bigl(-\frac{\Lambda}{3}\bigr)\alpha r_{+}}{2K_{+}}\epsilon^{2},\\
&\lim_{s\rightarrow\infty}\Omega^{2}e^{2K_{+}s} = 2K_{+},\label{asymptotic of Omega2}\\
&\lim_{s\rightarrow\infty}\frac{d}{ds}\log(\Omega^{2}) =-2K_{+} ,\\
&\lim_{s\rightarrow\infty}\frac{A}{\Omega^{2}} = \frac{e}{2K_{+}r_{+}^{2}},\\&
\lim_{s\rightarrow\infty}\frac{1}{\Omega^{2}}\frac{d\phi}{ds} = -\frac{\bigl(-\frac{\Lambda}{3}\bigr)\alpha\epsilon}{2K_{+}}.
\end{align}
Using the initial data of $\Omega^{2}$ \eqref{asymptotic of Omega2}, the coordinate transformation $s\rightarrow r$ is vaild near the event horizon. Since\begin{equation*}
\frac{d}{ds} = \frac{dr}{ds}\frac{d}{dr} = -\Omega^{2}\kappa^{-1}\frac{d}{dr},
\end{equation*}
we can write the ODE system $\eqref{emkg1}-\eqref{A equ}$ in terms of $r$\begin{align}
&\kappa^{-1}\frac{d\kappa}{dr} = r(\frac{d\phi}{dr})^{2}+q_{0}^{2}r\kappa^{2}\left(\frac{A}{\Omega^{2}}\right)^{2}\phi^{2},\label{req1}\\&\Omega^{2}\kappa^{-2} = 1-\frac{2\varpi}{r}+\frac{Q^{
2}}{r^{2}}+\left(-\frac{\Lambda}{3}\right)r^{2},\label{req2}\\&\frac{d\varpi}{dr} = \bigl(-\frac{\Lambda}{3}\bigr)\alpha\frac{r^{2}\phi^{2}}{2}+\frac{r^{2}(\frac{d\phi}{dr})^{2}}{2}\left(1-\frac{2\varpi}{r}+\frac{Q^{2}}{r^{2}}+\left(-\frac{\Lambda}{3}\right)r^{2}\right)+\frac{q_{0}^{2}A^{2}r^{2}\phi^{2}}{2\Omega^{2}}+q_{0}^{2}Qr\left(\frac{A}{\Omega^{2}}\right)\kappa\phi^{2},\label{req3}\\&\frac{dQ}{dr} = q_{0}^{2}r^{2}\kappa\frac{A}{\Omega^{2}}\phi^{2},\label{req5}\\&\frac{dA}{dr} = \frac{Q\kappa}{r^{2}}.\label{req6}
\\&\frac{d}{dr}\Bigl(r^{2}\Omega^{2}\kappa^{-2}\frac{d\phi}{dr}\Bigr)+r^{2}\kappa^{-1}\frac{d\kappa}{dr}\Omega^{2}\kappa^{-2}\frac{d\phi}{dr} = \bigl(-\frac{\Lambda}{3}\bigr)\alpha r^{2}\phi-\frac{q_{0}^{2}A^{2}r^{2}}{\Omega^{2}}\phi,\label{req4}
\end{align}
Then in terms of $r$, the initial conditions and gauge choice become\begin{align}
&\varpi(r_{+}) = M,\quad Q(r_{+}) = e,\quad \phi(r_{+}) = \epsilon\label{initial varpi}\\
&1-\frac{2M}{r_{+}}+\frac{e^{2}}{r_{+}^{2}}+\bigl(-\frac{\Lambda}{3}\bigr)r_{+}^{2} = 0,\\
&\kappa^{-1}(r_{+}) = 1-\frac{\bigl(-\frac{\Lambda}{3}\bigr)\alpha r_{+}}{2K_{+}}\epsilon^{2},\label{initial kappa}\\
&\frac{d\phi}{dr}(r_{+}) = \frac{\bigl(-\frac{\Lambda}{3}\bigr)\alpha\epsilon}{2K_{+}}\times\kappa(r_{+}),\\&
\lim_{r\rightarrow r_{+}}\frac{A}{\Omega^{2}} = \frac{e}{2K_{+}r_{+}^{2}},\\&
\frac{d\Omega^{2}}{dr}(r_{+}) = 2K_{+}\kappa(r_{+}).\label{initial dOmega}
\end{align}

It is obvious that we can deduce equations $\eqref{req1}-\eqref{req4}$ from $\eqref{emkg1}-\eqref{A equ}$ by the above method. One may could not find the equations corresponding to $\eqref{emkg1}$ and $\eqref{emkg2}$ in $\eqref{req1}-\eqref{req4}$. However, these two ODE systems actually are equivalent.
\begin{proposition}
The ODE system $\eqref{emkg1}-\eqref{A equ}$ is equivalent to the ODE system $\eqref{req1}-\eqref{req4}$.
\end{proposition}
\begin{proof}
We only need to use the ODE system $\eqref{req1}-\eqref{req6}$ to get the equations $\eqref{emkg1}$ and $\eqref{emkg2}$. Note that\begin{align*}
\frac{d}{ds}(-r\frac{dr}{ds}) &= \frac{d}{ds}(r\kappa^{-1}\Omega^{2}) = -\Omega^{2}\kappa^{-1}\frac{d}{dr}(r\kappa\kappa^{-2}\Omega^{2})\\&
 = -\Omega^{2}(\Omega^{2}\kappa^{-2})-\Omega^{2}r\kappa^{-1}\frac{d\kappa}{dr}\Omega^{2}\kappa^{-2}-\Omega^{2}r\frac{d}{dr}(\Omega^{2}\kappa^{-2})\\& = -\Omega^{2}(1-\frac{Q^{2}}{r^{2}})+r^{2}\Omega^{2}\bigl(\Lambda+\bigl(-\frac{\Lambda}{3}\bigr)\alpha\phi^{2}\bigr).
\end{align*}
The last identity is due to equations $\eqref{req1}-\eqref{req3}$. We can get the equation $\eqref{emkg2}$ by using the same method.
\end{proof}
\subsection{Boundary conditions}
\label{boundary}
In this section, we give the precise definition of various boundary conditions. Note in \cite{holzegel2014boundedness}, Holzegel and Warnick already gave a definition of Dirichlet, Neumann, and Robin boundary conditions. Since here we only consider the stationary spherically symmetric case, we can give the definition in a simplier way.
\begin{definition}
Let $\Delta = \sqrt{\frac{9}{4}+\alpha}$. We say a $C^{1}$ function $\phi$ on $\mathcal{M}$ obeys Dirichlet, Neumann or Robin boundary condition if the following hold:
\begin{enumerate}
  \item[(1)] Dirichlet:\begin{equation*}
  r^{\frac{3}{2}-\Delta}\phi\rightarrow 0,\quad r\rightarrow \infty.
  \end{equation*}
  \item[(2)] Neumann:\begin{equation*}
  r^{2\Delta+1}\frac{d}{dr}\left(r^{\frac{3}{2}-\Delta}\phi\right) = 0,\quad r\rightarrow\infty.
  \end{equation*}
  \item[(3)] Robin:\begin{equation*}
  r^{2\Delta+1}\frac{d}{dr}\left(r^{\frac{3}{2}-\Delta}\phi\right)+\beta r^{\frac{3}{2}-\Delta}\phi\rightarrow 0,\quad r\rightarrow\infty,
  \end{equation*}
  where $\beta$ is a real constant.
\end{enumerate}

\end{definition}

\subsection{Results on the Klein--Gordon equation on the Reissner--Nordström-AdS spacetime}
\label{pre of Reissner}
In this section, we recall the results about the stationary solutions to the Klein--Gordon equation on the Reissner--Nordström-AdS established in our companion work \cite{weihaozheng}.

Recall that for the Reissner--Nordström-AdS spacetime, we have\begin{equation}
Q\equiv e,\quad A  = A_{RN}= e\left(\frac{1}{r_{+}}-\frac{1}{r}\right),\quad\kappa \equiv 1.
\end{equation}
The Klein--Gordon equation on the Reissner--Nordström-AdS spacetime takes the form of\begin{equation}
g_{RN}^{\mu\nu}D_{\mu}D_{\nu}\phi = \bigl(-\frac{\Lambda}{3}\bigr)\alpha\phi.\label{linear klein gordon}
\end{equation}
The Klein--Gordon equation is parametrized by the parameters $(M,e,\Lambda,\alpha,q_{0})$. We call the parameters $(M,e,\Lambda,\alpha,q_{0})$ for the Klein--Gordon equation sub-extremal if the corresponding parameters $(M,e,\Lambda)$ are sub-extremal.

We can define the surface gravity $K_{+}$ of the black hole as follows: \begin{equation*}
2K_{+}: = \frac{d\Omega_{RN}^{2}}{dr}(r_{+}).
\end{equation*} 
For sub-extremal parameters, it is not hard to see that $\kappa$ is always positive. We call the parameters $(M,e,\Lambda)$ extremal if the function $\Omega_{RN}$ admits two same positive roots. For the extremal parameters, the surface gravity $K_{+} = 0$. We have\begin{align*}
&1-\frac{2M}{r_{+}}+\frac{e^{2}}{r_{+}^{2}}+(-\frac{\Lambda}{3})r_{+}^{2} =0,\\&
\frac{2M}{r_{+}^{2}}-\frac{2e^{2}}{r_{+}^{3}}+2(-\frac{\Lambda}{3})r_{+} = 0.
\end{align*}
We can solve $M$ in terms of $(r_{+},\Lambda)$ from above two equations:
\begin{equation*}
M_{0}: =r_{+}+2\bigl(-\frac{\Lambda}{3}\bigr)r_{+}^{3}.
\end{equation*}

We recall the following parameters transformation lemma in \cite{weihaozheng}:
\begin{lemma}\textup{\cite{weihaozheng}}
\label{parameter trans}
The parameters transformation $(M,e,\Lambda,\alpha,q_{0})\rightarrow(M,r_{+},\Lambda,\alpha,q_{0})$ is a regular transformation. The new parameters $(M,r_{+},\Lambda,\alpha,q_{0})$ satisfy the sub-extremal condition if and only if\begin{align}
&M\geq M_{e = 0}: = \frac{r_{+}}{2}\left(1+\bigl(-\frac{\Lambda}{3}\bigr)r_{+}^{2}\right),\label{subextremal1}\\&
M<M_{0} :=r_{+}+2\bigl(-\frac{\Lambda}{3}\bigr)r_{+}^{3}.\label{subextremal2}
\end{align}
In other words, for fixed $(r_{+},\Lambda,\alpha,q_{0})$, the sub-extremal condition can be achieved by making $M\in[M_{e = 0},M_{0})$.
\end{lemma}

In view of the above lemma, we can parametrize the Reissner--Nordström-AdS spacetime and Klein--Gordon equation on it by $(M,r_{+},\Lambda)$ and $(M,r_{+},\Lambda,\alpha,q_{0})$ respectively.

Now we are ready to state the main theorem established in \cite{weihaozheng}.

\begin{theorem}\textup{\cite{weihaozheng}}
\label{linear theory}
For any sub-extremal parameters $(M,r_{+},\Lambda,\alpha,q_{0})$ satisfying the Breitenlohner--Freedman bound $-\frac{9}{4}<\alpha<0$, there exists a stationary spherically symmetric solution $\phi(r)$ to the Klein--Gordon equation \eqref{linear klein gordon} such that\begin{align}
&\lim_{r\rightarrow r_{+}}\phi(r) = 1,\\
&\phi(r) = C_{D}u_{D}(r)+C_{N}u_{N}(r).
\end{align}
$\{u_{D}(r),u_{N}(r)\}$ is the local basis of the solutions to the stationary spherically symmetric Klein--Gordon equation \eqref{linear klein gordon}, with the following asymptotic behaviors:\begin{align}
&\lim_{r\rightarrow\infty}r^{\frac{3}{2}+\Delta}u_{D}(r) = 1,\quad \lim_{r\rightarrow\infty}r^{\frac{5}{2}+\Delta}\frac{du_{D}}{dr} = -\frac{3}{2}-\Delta,\\
&\lim_{r\rightarrow\infty}r^{\frac{3}{2}-\Delta}u_{N}(r) = 1,\quad \lim_{r\rightarrow\infty}r^{\frac{5}{2}-\Delta}\frac{du_{N}}{dr} = -\frac{3}{2}+\Delta.
\end{align} 
In other words, $u_{D}$ is a function on $\mathcal{M}$ satisfying the Dirichlet boundary condition and $u_{N}$ is a function on $\mathcal{M}$ satisfying the Neumann boundary condition. Moreover, there exist a large number $N$ and a small number $\delta$, such that the following estimates for $\{u_{D},u_{N}\}$ hold\begin{align}
&\left\vert u_{D}^{\prime}u_{N}-u_{D}u_{N}^{\prime}\right\vert>Cr^{-4},\quad r>N,\label{infty wronskian}\\
&\left\vert u_{D}^{\prime}u_{N}-u_{D}u_{N}^{\prime}\right\vert>\frac{C}{r-r_{+}},\quad r_{+}<r<r_{+}+\delta.\label{horizon wronskian}
\end{align}
\end{theorem}
\begin{remark}
Although $u_{D}$ and $u_{N}$ form a basis of the solutions to the stationary spherically symmetric Klein--Gordon equation on $(r_{+},\infty)$, they may blow up at the event horizon.
\end{remark}

\begin{theorem}\textup{\cite{weihaozheng}}
\label{linear hair theorem}
For Klein--Gordon equation $\eqref{linear klein gordon}$ with negative mass $\alpha$, let $-\frac{9}{4}<\alpha<0$ if the Dirichlet boundary condition is imposed; let $-\frac{9}{4}<\alpha<-\frac{5}{4}$ if the Neumann boundary condition is imposed. We have\begin{enumerate}
\item[(1)]\textup{(Large charge case)} For any given sub-extremal parameters $(M_{b},r_{+},\Lambda)$ and $\alpha$ within the above range, there exists a $q_{1}>0$, such that for any $\vert q_{0}\vert>q_{1}$, we can find $M_{e = 0}<M_{c} = M_{c}(M_{b},r_{+},\Lambda,\alpha,q_{0})<M_{b}<M_{0}$, such that there exists a stationary solution $\phi$ to $\eqref{linear klein gordon}$ with parameters $(M_{c},r_{+},\Lambda,\alpha,q_{0})$. Moreover, $\phi$ is a bounded function and can be extended continuously to the event horizon $r = r_{+}$.
\item[(2)]\textup{(Fixed charge case)} For any given parameters $(r_{+},\Lambda,\alpha,q_{0})$ satisfying
\begin{align}
&-\frac{9}{4}<\alpha<-\frac{3}{2}+\frac{q_{0}^{2}}{2\bigl(-\frac{\Lambda}{3}\bigr)},\label{hardy assumption}\\
&\bigl(-\frac{\Lambda}{3}\bigr)r_{+}^{2}>R_{0},\label{large condition}
\end{align}
there exists $M_{e=0}<M_{c}<M_{0}$ such that there exists a stationary solution $\phi$ to $\eqref{linear klein gordon}$ with parameters $(M_{c},r_{+},\Lambda,\alpha,q_{0})$. Moreover, $\phi$ is a bounded function and can be extended continuously to the event horizon $r = r_{+}$.
\end{enumerate}
\end{theorem}

\section{Main Results}
\label{Main results}
Now, we are ready to state the main theorem in this paper.
\begin{theorem}
\label{main theorem}
Imposing the Dirichlet (respectively Neumann) boundary condition for the scalar field $\phi$, letting $C_{DN} = 0$ ($-\frac{5}{4}$ for Neumann boundary conditions respectively), and using $(M,r_{+},\Lambda,\alpha,q_{0})$ as parameters, then we have the following two results:\begin{enumerate}
\item[(a)]\textup{(Large charge case)}Assume $(M_{b},r_{+},\Lambda)$ are given sub-extremal parameters and $\alpha$ is a fixed Klein--Gordon mass. Then there exists $q^{\max} = q^{\max}(M_{b},r_{+},\Lambda,\alpha)$, such that for any $\vert q_{0}\vert>q^{\max}$, there exists $\epsilon_{0}>0$ and $M_{e = 0}<M_{c}(\epsilon) = M_{c}(M_{b},r_{+},\Lambda,\alpha,q_{0},\epsilon)<M_{b}$ for all $0\leq\epsilon<\epsilon_{0}$, such that there exist stationary spherically symmetric black hole solutions $(\mathcal{M},g_{\epsilon},\phi_{\epsilon})$ to \eqref{original 1}-\eqref{original 4} with $\phi_{\epsilon}(r_{+}) = \epsilon$, bifurcating off the Reissner--Nordström-AdS $(\mathcal{M},g_{0},\phi_{0}\equiv0)$ with parameters $(M_{c}(0),r_{+},\Lambda)$.

\item[(b)]\textup{(Fixed charge condition)} For any fixed parameters $(r_{+},\Lambda,\alpha,q_{0})$ satisfying \begin{align}
&-\frac{9}{4}<\alpha<\min\left\{-\frac{3}{2}+\frac{q_{0}^{2}}{2\bigl(-\frac{\Lambda}{3}\bigr)},0\right\},\label{large1}\\
&\bigl(-\frac{\Lambda}{3}\bigr)r_{+}^{2}>R_{0}\label{large2},
\end{align}
there exists sub-extremal parameters $M_{e = 0}<M_{c}(\epsilon) = M_{c}(r_{+},\Lambda,\alpha,q_{0},\epsilon)<M_{0}$ such that there exist stationary spherically symmetric black hole solutions $(\mathcal{M},g_{\epsilon},\phi_{\epsilon})$ to \eqref{original 1}-\eqref{original 4} with $\phi_{\epsilon}(r_{+}) = \epsilon$, bifurcating off the Reissner--Nordström-AdS $(\mathcal{M},g_{0},\phi_{0}\equiv0)$ with parameters $(M_{c}(0),r_{+},\Lambda)$.
\end{enumerate}
Moreover, the bifurcating black hole solutions $(\mathcal{M},g_{\epsilon},\phi_{\epsilon})$ satisfy:\begin{enumerate}
\item \label{hairy1} For each $0\leq\epsilon<\epsilon_{0}$, $(\mathcal{M},g_{\epsilon},\phi_{\epsilon})$ are black hole solutions to the Einstein--Maxwell--(charged) Klein--Gordon equations under given boundary condition with the event horizon $\{r = r_{+}\}$. On the event horizon, $\phi_{\epsilon}$, $Q_{\epsilon}$, and $\varpi_{\epsilon}$ satisfy the conditions $\phi_{\epsilon}(r_{+}) = \epsilon$, $Q_{\epsilon}(r_{+}) = e(M_{c}(\epsilon),r_{+},\Lambda)$, and $\varpi_{\epsilon} = M_{c}(\epsilon)$.
\item \label{hairy3} For $\Omega_{\epsilon}^{2}$, we have the estimates\begin{align}
\left\vert\Omega_{\epsilon}^{2}\kappa_{\epsilon}^{-2}-\Omega_{RN}^{2}\right\vert&\lesssim\epsilon^{2},\\
\left\vert\frac{d\Omega_{\epsilon}^{2}\kappa_{\epsilon}^{-2}}{dr}-\frac{d\Omega_{RN}^{2}}{dr}\right\vert&\lesssim\epsilon^{2},
\end{align}
where $\Omega_{RN}^{2}$ is the Reissner--Nordström-AdS metric with parameters $(M_{c}(0),r_{+},\Lambda)$;
\item \label{hairy4} Moreover, the family is differentiable with respect to $\epsilon$ at $\epsilon = 0$, and $\lim_{\epsilon\rightarrow0}\epsilon^{-1}\phi_{\epsilon} = \hat{\phi}$, where $\hat{\phi}$ is a non-zero stationary solution to the (charged) Klein--Gordon equation on $(\mathcal{M},g_{0})$.
\end{enumerate}
\end{theorem}

One of the key steps toward proving Theorem \ref{main theorem} is to first show that there exists a boundary condition such that hairy black hole solutions under this boundary condition exist. More specifically, we have the following theorem:
\begin{theorem}
\label{thm:linear approximation}
For any given sub-extremal parameters $(M,e,\Lambda)$, charge $q_{0}\in\mathbb{R}$, and $-\frac{9}{4}<\alpha<0$, there exists $\epsilon_{0}>0$ small enough, for any $0\leq\epsilon<\epsilon_{0}$, there exists a non-zero pair $\left(\gamma(\epsilon),\beta(\epsilon)\right)$ such that if we impose the following boundary condition for $\phi$
\begin{equation}
\gamma(\epsilon) r^{2\Delta+1}\frac{d}{dr}\left(r^{\frac{3}{2}-\Delta}\phi\right)+\beta(\epsilon) r^{\frac{3}{2}-\Delta}\phi\rightarrow 0,\quad r\rightarrow\infty,\label{robin boundary condition}
\end{equation}
there exists a one-parameter family of static spherically symmetric spacetime $(\mathcal{M},g_{\epsilon},\phi_{\epsilon})$ satisfying \eqref{hairy1}-\eqref{hairy4} in Theorem \ref{main theorem}.
\end{theorem}
\begin{remark}
One can view the equation \eqref{robin boundary condition} as a linear combination of Dirichlet and Neumann boundary conditions. The proof of Theorem \ref{thm:linear approximation} is a standard application of linear approximation. However, as one can see from \eqref{robin boundary condition}, it is highly non-trival to show that for some parameters, $\beta$ actually vanishes.  
\end{remark}

\section{Proof of Theorem \ref{thm:linear approximation}}
\label{sec:proof of linear theorem}
In this section, we prove Theorem \ref{thm:linear approximation}. 
\paragraph{Outline of the proof of Theorem \ref{thm:linear approximation}}
By the discussion in Section \ref{Pre}, we have established the validity of the coordinate transformation $s\rightarrow r$ near the event horizon. Recall the compatible initial data on the event horizon \eqref{initial varpi}-\eqref{initial dOmega}:
\begin{align}
&\varpi(r_{+}) = M,\quad Q(r_{+}) = e,\quad \phi(r_{+}) = \epsilon\label{initial mass charge phi}\\
&1-\frac{2M}{r_{+}}+\frac{e^{2}}{r_{+}^{2}}+\bigl(-\frac{\Lambda}{3}\bigr)r_{+}^{2} = 0,\label{initial pOmega}\\
&\kappa^{-1}(r_{+}) = 1-\frac{\bigl(-\frac{\Lambda}{3}\bigr)\alpha r_{+}}{2K_{+}}\epsilon^{2},\label{initial pkappa}\\
&\frac{d\phi}{dr}(r_{+}) = \frac{\bigl(-\frac{\Lambda}{3}\bigr)\alpha\epsilon}{2K_{+}}\times\kappa(r_{+}),\label{initial dphi}\\&
\lim_{r\rightarrow r_{+}}\frac{A}{\Omega^{2}} = \frac{e}{2K_{+}r_{+}^{2}},\label{initial A}\\&
\frac{d\Omega^{2}}{dr}(r_{+}) = 2K_{+}\kappa(r_{+}).\label{initial pdOmega}
\end{align}
By constructing a proper norm (see \eqref{definition of norm}) and using a bootstrap argument, we can show the validity of transformation for all $s$ and extend the solution all the way to the spacetime infinity.

\subsection{Initial data and bootstrap assumption}
Consider the norm \begin{equation}
\Vert f\Vert_{N}(R) = \left\Vert r^{\frac{3}{2}-\Delta}f\right\Vert_{L^{\infty}[r_{+},R)}+\left\Vert r^{\frac{5}{2}-\Delta}\frac{df}{dr}\right\Vert_{L^{\infty}[r_{+},R)}+\left\Vert r^{\frac{3}{2}-\Delta}\Omega_{RN}^{2}\frac{d^{2}f}{dr^{2}}\right\Vert_{L^{\infty}[r_{+},R)},
\label{definition of norm}
\end{equation}
where $\Omega_{RN}^{2}$ takes the form of\begin{equation*}
\Omega_{RN}^{2} = 1-\frac{2M}{r}+\frac{e^{2}}{r^{2}}+\bigl(-\frac{\Lambda}{3}\bigr)r^{2}.
\end{equation*}
From the initial conditions \eqref{initial mass charge phi}-\eqref{initial pdOmega}, we have\begin{equation*}
\Vert \phi\Vert_{N}(r_{+}) = B\epsilon:=\left(r_{+}^{\frac{3}{2}-\Delta}+r_{+}^{\frac{5}{2}-\Delta}\frac{\bigl(-\frac{\Lambda}{3}\bigr)\alpha\kappa(r_{+})}{2K_{+}}\right)\epsilon.
\end{equation*}
Let $R>r_{+}$ be the largest real number such that for all $r\in(r_{+},R)$, we have\begin{align}
&\kappa(r_{+})\leq\kappa\leq \kappa(r_{+})+\epsilon,\label{b1}\\
&\Vert\phi\Vert_{N}(r)\leq(B+B_{0})\epsilon,\label{b2}\\
&\left\vert\frac{Ar^{2}}{\Omega^{2}}\right\vert\leq \frac{e}{2K_{+}}+\epsilon,\label{b3}
\end{align}
where $B$ and $\epsilon$ will be determined later. We also choose $\epsilon$ small enough such that we have $\kappa(r)\leq 2\kappa(r_{+})$ and $\frac{Ar^{2}}{\Omega^{2}}\leq \frac{e}{K_{+}}$ on $[r_{+},A]$.

In the later discussion, we use notation $C$ to denote constant depending only on the parameters $(M,e,r_{+},\alpha,q_{0})$.
\subsection{Estimate for $\kappa$}
We have the following estimate for $\kappa$.
\begin{lemma}
\label{lem Kap}
Under the bootstrap assumption $\eqref{b1}-\eqref{b3}$, $\kappa$ is increasing on $[r_{+},R]$ and $\kappa$ satisfies the estimate\begin{equation}
\kappa(r_{+})\leq\kappa(r)\leq \kappa(r_{+})+C(B+B_{0})^{2}\epsilon^{2},\quad \forall r_{+}\leq r\leq R.\label{kappa}
\end{equation}
\end{lemma}
\begin{proof}
Integrating the equation $\eqref{req1}$ for $\kappa$, we have\begin{equation}
\begin{aligned}
\log\kappa(r)-\log\kappa(r_{+}) &= \int_{r_{+}}^{r}\bar{r}\left(\frac{d\phi}{dr}(\bar{r})\right)^{2}+q_{0}^{2}\kappa^{2}\bar{r}\left(\frac{A}{\Omega^{2}}\right)^{2}\phi^{2}\mathrm{d}\bar{r}\\&\leq C(B+B_{0})^{2}\epsilon^{2}\int_{r_{+}}^{r}\bar{r}^{-4+2\Delta}\mathrm{d}\bar{r}+C(B+B_{0})^{2}\epsilon^{2}\int_{r}^{r}\bar{r}^{-5+2\Delta}\mathrm{d}\bar{r}\\&
\leq C(B+B_{0})^{2}\epsilon^{2}.
\end{aligned}
\end{equation}
Thus we have\begin{equation*}
\kappa(r_{+})\leq \kappa(r)\leq \kappa(r_{+})e^{C(B+B_{0})^{2}\epsilon^{2}}\leq \kappa(r_{+})+C(B+B_{0})^{2}\epsilon^{2},\quad \forall r_{+}\leq r\leq R.
\end{equation*}
\end{proof}
\subsection{Estimate for $Q$}
As one can see from the equations \eqref{emkg1}-\eqref{A equ}, for the uncharged scalar field, $Q\equiv e$. However, due to the presence of the scalar field charge $q_{0}$, $Q$ is not necessary bounded. Specifically, we have the following lemma:
\begin{lemma}
\label{Q}
Under the bootstrap assumption $\eqref{b1}-\eqref{b3}$, we have\begin{equation}
\left\vert Q(r)-e\right\vert\leq C(B+B_{0})^{2}\epsilon^{2}\left\vert r^{-2+2\Delta}-r_{+}^{-2+2\Delta}\right\vert,\quad r_{+}\leq r\leq R.
\end{equation}
\end{lemma}
\begin{proof}
Directly integrating $\eqref{req5}$, we have\begin{equation*}
\left\vert Q(r)-e\right\vert\leq C(B+B_{0})^{2}\epsilon^{2}\int_{r_{+}}^{r}\bar{r}^{-3+2\Delta}\mathrm{d}\bar{r}\leq C(B+B_{0})^{2}\epsilon^{2}\vert r^{-2+2\Delta}-r_{+}^{-2+2\Delta}\vert.
\end{equation*}
\end{proof}
\subsection{Estimate for Vaidya Mass $\varpi$}
\begin{lemma}
\label{varpi}
Under the bootstrap assumption $\eqref{b1}-\eqref{b3}$, we have\begin{equation}
\left\vert\varpi(r)-M\right\vert\leq C(B+B_{0})^{2}\epsilon^{2}r^{2\Delta-1}(r-r_{+}),\quad \forall r_{+}\leq r\leq A.
\end{equation}
\end{lemma}
\begin{proof}
First, we estimate $e^{\int_{r_{+}}^{r}\bar{r}\left(\bigl(\frac{d\phi}{dr}\bigr)^{2}+\kappa^{2}q_{0}^{2}\bigl(\frac{A}{\Omega^{2}}\bigr)\phi^{2}\right)d\bar{r}}$ as follows:\begin{equation*}
e^{\int_{r_{+}}^{r}\bar{r}\left(\bigl(\frac{d\phi}{dr}\bigr)^{2}+\kappa^{2}q_{0}^{2}\bigl(\frac{A}{\Omega^{2}}\bigr)\phi^{2}\right)d\bar{r}}\leq e^{C(B+B_{0})^{2}\epsilon^{2}}\leq 1+C(B+B_{0})^{2}\epsilon^{2}.
\end{equation*}

By using the integrating factor $e^{\int_{r_{+}}^{r}\bar{r}\left(\bigl(\frac{d\phi}{dr}\bigr)^{2}+\kappa^{2}q_{0}^{2}\bigl(\frac{A}{\Omega^{2}}\bigr)\phi^{2}\right)d\bar{r}}$, we can write the equation for $\varpi$ $\eqref{req3}$ as
\begin{align*}
&e^{-\int_{r_{+}}^{r}\bar{r}\left(\bigl(\frac{d\phi}{dr}\bigr)^{2}+\kappa^{2}q_{0}^{2}\bigl(\frac{A}{\Omega^{2}}\bigr)\phi^{2}\right)d\bar{r}}\frac{d}{dr}\left(e^{\int_{r_{+}}^{r}\bar{r}\left(\bigl(\frac{d\phi}{dr}\bigr)^{2}+\kappa^{2}q_{0}^{2}\bigl(\frac{A}{\Omega^{2}}\bigr)\phi^{2}\right)d\bar{r}}\varpi\right) \\=&\bigl(-\frac{\Lambda}{3}\bigr)\alpha\frac{r^{2}\phi^{2}}{2}+
q_{0}^{2}Q\kappa r\left(\frac{A}{\Omega^{2}}\right)\phi^{2}+\frac{r^{2}}{2}\left(\bigl(\frac{d\phi}{dr}\bigr)^{2}+q_{0}^{2}\kappa^{2}\left(\frac{A}{\Omega^{2}}\right)^{2}\phi^{2}\right)\left(1+\frac{Q^2}{r^{2}}+\bigl(-\frac{\Lambda}{3}\bigr)r^{2}\right).
\end{align*}
Directly integrating the above equation, we have\begin{align*}
\left\vert e^{\int_{r_{+}}^{r}\bar{r}\left(\left(\frac{d\phi}{dr}\right)^{2}+\kappa^{2}q_{0}^{2}\left(\frac{A}{\Omega^{2}}\right)^{2}\phi^{2}\right)\mathrm{d}\bar{r}}\varpi(r)-M\right\vert\leq C(B+B_{0})^{2}\epsilon^{2}r^{2\Delta-1}(r-r_{+}).
\end{align*}
Hence we have\begin{align*}
\left\vert\varpi-M\right\vert\leq& \left\vert e^{-\int_{r_{+}}^{r}\bar{r}\left(\left(\frac{d\phi}{dr}\right)^{2}+q_{0}^{2}\kappa^{2}\left(\frac{A}{\Omega^{2}}\right)^{2}\phi^{2}\right)\mathrm{d}\bar{r}}M-M\right\vert+C(B+B_{0})^{2}\epsilon^{2}r^{2\Delta-1}(r-r_{+})\\\leq&
C(B+B_{0})^{2}\epsilon^{2}r^{2\Delta-1}(r-r_{+}).
\end{align*}
\end{proof}
\subsection{Estimate for $\frac{Ar^{2}}{\Omega^{2}}$}
For $\frac{Ar^{2}}{\Omega^{2}}$, we have the following estimate.
\begin{lemma}
Under the bootstrap assumption $\eqref{b1}-\eqref{b3}$, we have\begin{equation}
\label{estimate for ren A}
\left\vert \frac{Ar^{2}}{\Omega^{2}\kappa^{-2}}-\frac{A_{RN}r^{2}}{\Omega_{RN}^{2}}\right\vert\leq C(B+B_{0})^{2}\epsilon^{2}.
\end{equation}
\end{lemma}
\begin{proof}
We can write $\Omega_{RN,M}^{2}$ as \begin{equation*}
\Omega_{RN}^{2} = (r-r_{+})P(r),
\end{equation*}
where $P$ is a rational function of degree $1$ with a lower bound $C$ on $(r_{+},\infty)$. Then we have\begin{equation}
\begin{aligned}
\frac{Ar^{2}}{\Omega^{2}\kappa^{-2}}-\frac{A_{RN}r^{2}}{\Omega_{RN}^{2}} =& \frac{Ar^{2}}{(r-r_{+})\left(P(r)+\frac{(Q-e)^{2}+2e(Q-e)}{r^{2}(r-r_{+})}+\frac{2M-2\varpi}{r(r-r_{+})}\right)}-\frac{Ar^{2}}{(r-r_{+})P(r)}\\& +
\frac{Ar^{2}}{(r-r_{+})P(r)}-\frac{A_{RN}r^{2}}{(r-r_{+})P(r)}\\=&
\frac{Ar^{2}}{r-r_{+}}\frac{\frac{(Q-e)^{2}+2e(Q-e)}{r^{2}(r-r_{+})}+\frac{2M-2\varpi}{r(r-r_{+})}}{P(r)\left(P(r)+\frac{(Q-e)^{2}+2e(Q-e)}{r^{2}(r-r_{+})}+\frac{2M-2\varpi}{r(r-r_{+})}\right)}+\frac{(A-A_{RN})r^{2}}{(r-r_{+})P(r)}.
\end{aligned}
\label{decom of A}
\end{equation}
Using the equation $\eqref{req6}$, we have\begin{equation}
\label{req6 for pertubate version}
\frac{d\left(A-A_{RN}\right)}{dr} = \frac{Q\kappa}{r^{2}}-\frac{e}{r^{2}} = \frac{(Q-e)\kappa}{r^{2}}+\frac{e(\kappa-1)}{r^{2}}.
\end{equation}
Directly integrating the above equation, we have\begin{equation}
\vert A-A_{RN}\vert\leq\int_{r_{+}}^{r}\frac{\vert Q-e\vert\kappa}{\bar{r}^{2}}+\frac{\vert e\vert(\kappa-1)}{\bar{r}^{2}}\mathrm{d}\bar{r}\leq C(B+B_{0})^{2}\epsilon^{2}\left(\frac{1}{r_{+}}-\frac{1}{r}\right).\label{estimate for A}
\end{equation}
Then by \eqref{decom of A}, we can get the estimate\begin{equation*}
\left\vert \frac{Ar^{2}}{\Omega^{2}\kappa^{-2}}-\frac{A_{RN}r^{2}}{\Omega_{RN}^{2}}\right\vert\leq C(B+B_{0})^{2}\epsilon^{2}
\end{equation*}
\end{proof}
\subsection{Estimate for $\phi$ and close the bootstrap estimate}
Now we are ready to prove Theorem $\ref{thm:linear approximation}$.
\begin{proof}
Considering the equation for $\phi$, we can write
\begin{equation}
\label{linear and error version of the main equation}
\begin{aligned}
-\frac{d}{dr}\left(r^{2}\Omega^{2}_{RN}\frac{d\phi}{dr}\right)+\left(\bigl(-\frac{\Lambda}{3}\bigr)\alpha-\frac{q_{0}^{2}A_{RN}^{2}}{\Omega_{RN}^{2}}\right)r^{2}\phi =&
\frac{d}{dr}\left((2M-2\varpi)r\frac{d\phi}{dr}\right)+\frac{d}{dr}\left(\left(Q^{2}-e^{2}\right)\frac{d\phi}{dr}\right)\\&+r^{2}\kappa^{-1}\frac{d\kappa}{dr}\Omega^{2}\kappa^{-2}\frac{d\phi}{dr}+q_{0}^{2}r^{2}\left(\frac{A^{2}}{\Omega^{2}}-\frac{A_{RN}^{2}}{\Omega_{RN}^{2}}\right)\phi.
\end{aligned}
\end{equation}
Let
\begin{align*}
l[\phi] =& -\frac{d}{dr}\left(r^{2}\Omega^{2}_{RN}\frac{d\phi}{dr}\right)+\left(\bigl(-\frac{\Lambda}{3}\alpha\bigr)-\frac{q_{0}^{2}A_{RN}^{2}}{\Omega_{RN}^{2}}\right)r^{2}\phi,\\
e[\phi] =& \frac{d}{dr}\left((2M-2\varpi)r\frac{d\phi}{dr}\right)+\frac{d}{dr}\left(\left(Q^{2}-e^{2}\right)\frac{d\phi}{dr}\right)\\&
+r^{2}\kappa^{-1}\frac{d\kappa}{dr}\Omega^{2}\kappa^{-2}\frac{d\phi}{dr}+q_{0}^{2}r^{2}\left(\frac{A^{2}}{\Omega^{2}}-\frac{A_{RN}^{2}}{\Omega_{RN}^{2}}\right)\phi.
\end{align*}
We can view $l[\phi]$ as the main part of the equation \eqref{linear and error version of the main equation} and $e[\phi]$ as the error terms. From Theorem \ref{linear theory}, we know that $\phi = C_{D}u_{D}+C_{N}u_{N}$ is the solution for $l[\phi] = 0$ with $\phi$ regular at the event horizon $\phi(r_{+}) = 1$. Let\begin{align*}
&\phi(r) = C_{D}(r)u_{D}(r)+C_{N}(r)u_{N}(r),\\&
C_{D}(r_{+}) = \epsilon C_{D},\quad C_{N}(r_{+}) = \epsilon C_{N}
\end{align*}
be the solution of $l[\phi] = e[\phi]$ with $\phi(r_{+}) = \epsilon$.
Then we can get\begin{align}
&C_{D}^{\prime}(r)u_{D}(r)+C_{N}^{\prime}(r)u_{N}(r) = 0,\label{relation1}\\&
C_{D}^{\prime}(r)u_{D}^{\prime}(r)+C_{N}^{\prime}(r)u_{N}^{\prime}(r) =-\frac{1}{r^{2}\Omega_{RN}^{2}}e[\phi].\label{relation2} 
\end{align}
From $\eqref{relation1}$ and $\eqref{relation2}$, we can solve $C_{D}^{\prime}$ and $C_{N}^{\prime}$\begin{align}
&C_{D}^{\prime}(r) = -\frac{u_{N}}{u_{D}^{\prime}u_{N}-u_{D}u_{N}^{\prime}}\frac{e[\phi]}{r^{2}\Omega^{2}_{RN}},\label{Cd}\\&
C_{N}^{\prime}(r) = \frac{u_{D}}{u_{D}^{\prime}u_{N}-u_{D}u_{N}^{\prime}}\frac{e[\phi]}{r^{2}\Omega^{2}_{RN}}.\label{Cn}
\end{align}
Integrating the equaiton $\eqref{Cd}$ for $C_{D}(r)$, we have\begin{equation}
\begin{aligned}
C_{D}(r)-\epsilon C_{D} &= \int_{r_{+}}^{r}\frac{-u_{N}}{u_{D}^{\prime}u_{N}-u_{D}u_{N}^{\prime}}\frac{e[\phi]}{\bar{r}^{2}\Omega^{2}_{RN}}\mathrm{d}\bar{r}\\
 &= \int_{r_{+}}^{r}\frac{-u_{N}}{u_{D}^{\prime}u_{N}-u_{D}u_{N}^{\prime}}\left(\kappa^{-1}\frac{d\kappa}{dr}\frac{d\phi}{dr}\frac{\Omega^{2}\kappa^{-2}}{\Omega^{2}_{RN}}\right)+\frac{-u_{N}}{u_{D}^{\prime}u_{N}-u_{D}u^{\prime}_{N}}\frac{q_{0}^{2}\left(\frac{A^{2}}{\Omega^{2}}-\frac{A_{RN}^{2}}{\Omega_{RN}^{2}}\right)}{\Omega_{RN}^{2}}\\&
 +\frac{-u_{N}}{u_{D}^{\prime}u_{N}-u_{D}u_{N}^{\prime}}\frac{\frac{d}{dr}\left((2M-2\varpi)\bar{r}\frac{d\phi}{dr}\right)}{\bar{r}^{2}\Omega_{RN}^{2}}+\frac{-u_{N}}{u_{D}^{\prime}u_{N}-u_{N}^{\prime}u_{D}}\frac{\frac{d}{dr}\left(\left(Q^{2}-e^{2}\right)\frac{d\phi}{dr}\right)}{\bar{r}^{2}\Omega_{RN}^{2}}\mathrm{d}\bar{r}.
\end{aligned}
\label{integral Cd}
\end{equation}
Then by $\eqref{infty wronskian}$ and $\eqref{horizon wronskian}$, we have\begin{align}
&\left\vert\frac{u_{N}}{u_{D}^{\prime}u_{N}-u_{D}u_{N}^{\prime}}\right\vert\leq C(r-r_{+})\left\vert\log(r-r_{+})\right\vert,\quad r\rightarrow r_{+},\\&
\left\vert\frac{u_{N}}{u_{D}^{\prime}u_{N}-u_{D}u_{N}^{\prime}}\right\vert\leq Cr^{\frac{5}{2}-\Delta},\quad r\rightarrow\infty,\\&
\frac{1}{r-r_{+}}\left\vert\frac{u_{N}}{u_{D}^{\prime}u_{N}-u_{D}u_{N}^{\prime}}\right\vert\leq C\max\{-\log(r-r_{+}),r^{\frac{3}{2}-\Delta}\},\quad r_{+}<r<\infty.\label{ren for wron}
\end{align}

Since we have estimate for $\varpi$ in Lemma $\ref{varpi}$ and estimate for $Q$ in Lemma $\ref{Q}$, by $\eqref{req2}$, we have\begin{equation*}
\frac{\Omega^{2}\kappa^{-2}}{\Omega_{RN}^{2}}<2.
\end{equation*}
For the first term in $\eqref{integral Cd}$, we have\begin{equation}
\begin{aligned}
\left\vert\int_{r_{+}}^{r}\frac{-u_{N}}{u_{D}^{\prime}u_{N}-u_{D}u_{N}^{\prime}}\left(\kappa^{-1}\frac{d\kappa}{dr}\frac{d\phi}{dr}\frac{\Omega^{2}\kappa^{-2}}{\Omega_{RN}^{2}}\right)\mathrm{d}\bar{r}\right\vert
\leq& C(B+B_{0})^{3}\epsilon^{3}\int_{r_{+}}^{r}\bar{r}^{-4+2\Delta}\mathrm{d}\bar{r}\\\leq&
C(B+B_{0})^{3}\epsilon^{3}\left(r_{+}^{-3+2\Delta}-r^{-3+2\Delta}\right).
\end{aligned}\label{1}
\end{equation}
By \eqref{kappa}, \eqref{estimate for A}, and \eqref{estimate for ren A}, we have\begin{equation*}
\left\vert\frac{A^{2}}{\Omega^{2}}-\frac{A_{RN}^{2}}{\Omega_{RN}^{2}}\right\vert\leq C\epsilon^{2}r^{-2}\left(\frac{1}{r_{+}}-\frac{1}{r}\right) 
\end{equation*}
Then we can estimate the second term in $\eqref{integral Cd}$\begin{equation}
\label{2}
\begin{aligned}
&\int_{r_{+}}^{r}\left\vert\frac{u_{N}}{u_{D}^{\prime}u_{N}-u_{D}u^{\prime}_{N}}\right\vert\frac{q_{0}^{2}\left\vert\frac{A^{2}}{\Omega^{2}}-\frac{A_{RN}^{2}}{\Omega_{RN}^{2}}\right\vert\vert\phi\vert}{\Omega_{RN}^{2}}\mathrm{d}\bar{r}\leq C(B+B_{0})\epsilon^{3}\int_{r_{+}}^{r}\bar{r}^{-3}\mathrm{d}\bar{r}\leq C(B+B_{0})\epsilon^{3}\left(r_{+}^{-2}-r^{-2}\right).
\end{aligned}
\end{equation}
For the third term in $\eqref{Cd}$, we have\begin{equation}
\begin{aligned}
&\int_{r_{+}}^{r}\left\vert\frac{u_{N}}{u_{D}^{\prime}u_{N}-u_{D}u_{N}^{\prime}}\right\vert\left\vert\frac{\frac{d}{dr}\left((2\varpi-2M)\bar{r}\frac{d\phi}{dr}\right)}{\bar{r}^{2}\Omega_{RN}^{2}}\right\vert\mathrm{d}\bar{r}\\\leq&
\int_{r_{+}}^{r}\left\vert\frac{u_{N}}{u_{D}^{\prime}u_{N}-u_{D}u_{N}^{\prime}}\right\vert\left(\frac{\vert2\varpi-2M\vert\left\vert\frac{d^{2}\phi}{dr^{2}}\right\vert}{\bar{r}\Omega_{RN}^{2}}+\frac{\vert2\varpi-2M\vert\left\vert\frac{d\phi}{dr}\right\vert}{\bar{r}^{2}\Omega_{RN}^{2}}+\frac{2\left\vert\frac{d\varpi}{dr}\right\vert\left\vert\frac{d\phi}{dr}\right\vert}{\bar{r}\Omega_{RN}^{2}}\right)\mathrm{d}\bar{r}\\\leq&
C(B+B_{0})^{3}\epsilon^{3}(r_{+}^{-3+2\Delta}-r^{-3+2\Delta})+C\int_{r_{+}}^{r}\left\vert\frac{u_{N}}{u_{D}^{\prime}u_{N}-u_{D}u_{N}^{\prime}}\right\vert\frac{\vert\varpi-M\vert\left\vert\frac{d^{2}\phi}{dr^{2}}\right\vert}{\bar{r}^{2}(\bar{r}-r_{+})}\mathrm{d}\bar{r}\\&+C(B+B_{0})\epsilon\int_{r_{+}}^{r}\left\vert\frac{u_{N}}{u_{D}^{\prime}u_{N}-u_{D}u_{N}^{\prime}}\right\vert \bar{r}^{-\frac{9}{2}+\Delta}\frac{\left\vert\frac{d\varpi}{dr}\right\vert}{\bar{r}-r_{+}}\mathrm{d}\bar{r},
\end{aligned}
\end{equation}
where we used the bootstrap estimate for $\varpi$ in Lemma $\ref{varpi}$. Using the Wronskian estimate \eqref{ren for wron}, we have\begin{align*}
&\int_{r_{+}}^{r}\left\vert\frac{u_{N}}{u_{D}^{\prime}u_{N}-u_{D}u_{N}^{\prime}}\right\vert\frac{\vert\varpi-M\vert\left\vert\frac{d^{2}\phi}{dr^{2}}\right\vert}{\bar{r}^{2}(\bar{r}-r_{+})}\mathrm{d}\bar{r}
\leq C(B+B_{0})^{3}\epsilon^{3}\min\{(r-r_{+})\vert\log(r-r_{+})\vert,1\},\\
&\int_{r_{+}}^{r}\left\vert\frac{u_{N}}{u_{D}^{\prime}u_{N}-u_{D}u_{N}^{\prime}}\right\vert \bar{r}^{-\frac{9}{2}+\Delta}\frac{\left\vert\frac{d\varpi}{dr}\right\vert}{\bar{r}-r_{+}}\mathrm{d}\bar{r}\leq C(B+B_{0})^{2}\epsilon^{2}\min\{(r-r_{+})\vert\log(r-r_{+})\vert,1\}.
\end{align*}
Hence we can estimate the third term in $\eqref{Cd}$\begin{equation}
\label{3}
\begin{aligned}
&\int_{r_{+}}^{r}\left\vert\frac{u_{N}}{u_{D}^{\prime}u_{N}-u_{D}u_{N}^{\prime}}\right\vert\left\vert\frac{\frac{d}{dr}\left((2\varpi-2M)\bar{r}\frac{d\phi}{dr}\right)}{\bar{r}^{2}\Omega_{RN}^{2}}\right\vert\mathrm{d}\bar{r}\\\leq&
C(B+B_{0})^{3}\epsilon^{3}\left(r_{+}^{-3+2\Delta}-r^{-3+2\Delta}+\min\{(r-r_{+})\vert\log(r-r_{+})\vert,1\}\right).
\end{aligned}
\end{equation}

Similarly, use the bootstrap estimate for $Q$ in Lemma $\ref{Q}$ and Wronskian estimate $\eqref{ren for wron}$, we have\begin{equation}
\label{4}
\int_{r_{+}}^{r}\left\vert\frac{u_{N}}{u_{D}^{\prime}u_{N}-u_{D}u_{N}^{\prime}}\right\vert\left\vert\frac{\left\vert\frac{d}{dr}\left(\left(Q^{2}-e^{2}\right)\frac{d\phi}{dr}\right)\right\vert}{\bar{r}^{2}\Omega_{RN}^{2}}\right\vert\mathrm{d}\bar{r}\leq C(B+B_{0})^{3}\epsilon^{3}\min\{(r-r_{+})\vert\log(r-r_{+})\vert,1\}.
\end{equation}
Combine $\eqref{1}$, $\eqref{2}$, $\eqref{3}$ and $\eqref{4}$, we get\begin{equation}
\vert C_{D}(r)-\epsilon C_{D}\vert\leq C(B+B_{0})^{3}\epsilon^{3}\min\{(r-r_{+})\vert\log(r-r_{+})\vert,1\},\quad \forall r_{+}\leq r<A.\label{bound for CD}
\end{equation}

Next, we estimate $C_{N}$. Similarly, we can integrate the equation $\eqref{Cn}$\begin{equation}
\begin{aligned}
C_{N}(r)-\epsilon C_{N} &= \int_{r_{+}}^{r}\frac{u_{D}}{u_{D}^{\prime}u_{N}-u_{D}u_{N}^{\prime}}\frac{e[\phi]}{\bar{r}^{2}\Omega^{2}_{RN}}\mathrm{d}\bar{r}\\
 &= \int_{r_{+}}^{r}\frac{u_{D}}{u_{D}^{\prime}u_{N}-u_{D}u_{N}^{\prime}}\left(\kappa^{-1}\frac{d\kappa}{dr}\frac{d\phi}{dr}\frac{\Omega^{2}\kappa^{-2}}{\Omega^{2}_{RN}}\right)\mathrm{d}\bar{r}+\frac{u_{D}}{u_{D}^{\prime}u_{N}-u_{D}u_{N}^{\prime}}\frac{\left(\frac{A^{2}}{\Omega^{2}}-\frac{A_{RN}^{2}}{\Omega_{RN}^{2}}\right)}{\Omega_{RN}^{2}}\\&+
\frac{u_{D}}{u_{D}^{\prime}u_{N}-u_{D}u^{\prime}_{N}}\frac{\frac{d}{dr}\left((2M-2\varpi)\bar{r}\frac{d\phi}{dr}\right)}{\bar{r}^{2}\Omega_{RN}^{2}}+\frac{u_{D}}{u_{D}^{\prime}u_{N}-u_{D}u_{N}^{\prime}}\frac{\frac{d}{dr}\left((Q^{2}-e^{2})\frac{d\phi}{dr}\right)}{\bar{r}^{2}\Omega_{RN}^{2}}\mathrm{d}\bar{r}.
\end{aligned}
\label{integral Cn}
\end{equation}
By $\eqref{infty wronskian}$ and $\eqref{horizon wronskian}$, we have\begin{align*}
&\frac{u_{D}}{u_{D}^{\prime}u_{N}-u_{D}u_{N}^{\prime}}\leq Cr^{\frac{5}{2}+\Delta},\quad r\rightarrow \infty,\\&
\frac{u_{D}}{u_{D}^{\prime}u_{N}-u_{D}u_{N}^{\prime}}\leq C(r-r_{+})\vert\log(r-r_{+})\vert.
\end{align*}
We can estimate all four terms in $\eqref{integral Cn}$ similar as in $\eqref{integral Cd}$. Thus we have\begin{equation}
\left\vert C_{N}(r)-\epsilon C_{N}\right\vert\leq C(B+B_{0})^{3}\epsilon^{3}r^{2\Delta}\min\{(r-r_{+})\vert\log(r-r_{+})\vert,1\},\ \forall r_{+}\leq r<A.\label{bound for Cn}
\end{equation}
If $A<\infty$, then combining $\eqref{bound for CD}$ and $\eqref{bound for Cn}$,  for $r_{+}<r<A$, we have\begin{equation}
\begin{aligned}
\left\vert r^{\frac{3}{2}-\Delta}\phi\right\vert\leq& \left \vert r^{\frac{3}{2}-\Delta}(C_{D}(r)-\epsilon C_{D})u_{D}\right\vert+\left\vert r^{\frac{3}{2}-\Delta}(C_{N}(r)-\epsilon C_{D})u_{N}\right\vert+C_{0}\epsilon\\\leq& C(B+B_{0})^{3}\epsilon^{3}+C_{0}\epsilon.
\end{aligned}
\end{equation}
For $\left\vert r^{\frac{5}{2}-\Delta}\frac{d\phi}{dr}\right\vert$ term, we have\begin{equation}
\left\vert r^{\frac{5}{2}-\Delta}\frac{d\phi}{dr}\right\vert = \left\vert C_{D}(r)u_{D}^{\prime}+C_{N}(r)u_{N}^{\prime}\right\vert\leq C(B+B_{0})^{3}\epsilon^{3}+\epsilon C_{0}.
\end{equation}
For $\vert r^{\frac{3}{2}-\Delta}\Omega_{RN}^{2}\frac{d^{2}\phi}{dr^{2}}\vert$, we have\begin{equation}
\begin{aligned}
\left\vert r^{\frac{3}{2}-\Delta}\Omega_{RN}^{2}\frac{d^{2}\phi}{dr^{2}}\right\vert =&r^{\frac{3}{2}-\Delta}\Omega_{RN}^{2}\left\vert C^{\prime}_{D}(r)u_{D}^{\prime}+C^{\prime}_{N}(r)u_{N}^{\prime}+C_{D}(r)u_{D}^{\prime\prime}+C_{N}(r)u_{N}^{\prime\prime}\right\vert\\\leq&
r^{\frac{3}{2}-\Delta}\Omega_{RN}^{2}\vert C_{D}^{\prime}(r)u_{D}^{\prime}+C_{N}^{\prime}(r)u_{N}^{\prime}\vert+r^{\frac{3}{2}-\Delta}\Omega_{RN}^{2}\epsilon\vert C_{D}u_{D}^{\prime\prime}+C_{N}u_{N}^{\prime\prime}\vert\\&+r^{\frac{3}{2}-\Delta}\Omega_{RN}^{2}\vert C_{D}(r)-\epsilon C_{D}\vert\vert u_{D}^{\prime\prime}\vert+r^{\frac{3}{2}-\Delta}\Omega_{RN}^{2}\vert C_{N}(r)-\epsilon C_{N}\vert\vert u_{N}^{\prime\prime}\vert.
\end{aligned}
\end{equation}
Using the structure of the equation, we can express $u_{D}^{\prime\prime}$ and $u_{N}^{\prime\prime}$ by lower order terms. Hence we have\begin{equation}
r^{\frac{3}{2}-\Delta}\Omega_{RN}^{2}\left\vert\frac{d^{2}\phi}{dr^{2}}\right\vert\leq C(B_{0}+B)^{3}\epsilon^{3}+C_{0}\epsilon.
\end{equation}

Thus we can choose $\epsilon$ small and $B_{0} = 2C_{0}$ such that we have\begin{align*}
&\kappa(r_{+})\leq \kappa\leq \kappa(r_{+})+\frac{1}{2}\epsilon,\\
&\Vert \phi\Vert_{N}(A)\leq (B+C_{0})\epsilon,\\&
\left\vert\frac{Ar^{2}}{\Omega^{2}}\right\vert\leq\frac{e}{2K_{+}}+\frac{\epsilon}{2},
\end{align*}
which contradicts to the fact that $A<\infty$. So we have $A = \infty$.
\end{proof}

\section{Proof of Theorem \ref{main theorem} under Dirichlet boundary conditions}
\label{proof:Dirichlet}
In this section, we will prove Theorem \ref{main theorem} under Dirichlet boundary conditions. The construction of the bifurcating black holes can be reduced to the construction of non-trivial solutions to \eqref{req1}-\eqref{req4} with $\phi(r_{+}) = \epsilon$ under Dirichlet boundary conditions. The proof is based on the fixed point argument and Theorem \ref{linear hair theorem}. Without loss of generality, we assume $e>0$ in this section.
\subsection{Function space}
In this section, we set up the function space we will work in. Define the norm $\Vert\cdot\Vert_{X}$
\begin{equation*}
\Vert f\Vert_{X} = \Vert r^{\frac{3}{2}+\Delta}f\Vert_{\infty}+\left\Vert r^{\frac{5}{2}+\Delta}\frac{df}{dr}\right\Vert_{\infty}+\left\Vert r^{\frac{3}{2}+\Delta}\Omega_{RN}^{2}\frac{d^{2}f}{dr^{2}}\right\Vert_{\infty}.
\end{equation*}
Let space $X$ be the completion of all $f\in C^{2}(r_{+},\infty)$ with $\Vert f\Vert_{X}<\infty$. Clearly from the definition of $X$, for any $f\in X$, we have\begin{equation*}
\lim_{r\rightarrow\infty}\left\vert r^{\frac{3}{2}-\Delta}f\right\vert = 0.
\end{equation*}
Let $X_{\delta}: = B_{\delta}(X)$ be the ball of radius $\delta$ in $X$. $\delta$ will be chosen to be small later.
\subsection{Set-up for the fixed point argument}
Given $f\in X_{\delta}$, we can consider the equations \eqref{req1}-\eqref{req6} with $\phi$ replaced by $f$:\begin{align}
&\kappa^{-1}\frac{d\kappa}{dr} = r\left(\frac{df}{dr}\right)^{2}+q_{0}^{2}r\kappa^{2}\left(\frac{A}{\Omega^{2}}\right)^{2}f^{2},\label{fixpointeq1}\\&
\Omega^{2}\kappa^{-2} = 1-\frac{2\varpi}{r}+\frac{Q^{2}}{r^{2}}+\bigl(-\frac{\Lambda}{3}\bigr)r^{2}\label{fixpointeq2},\\&
\frac{d\varpi}{dr} = \bigl(-\frac{\Lambda}{3}\bigr)\alpha\frac{r^{2}f^{2}}{2}+\frac{r^{2}(\frac{df}{dr})^{2}}{2}\Omega^{2}\kappa^{-2}+\frac{q_{0}^{2}A^{2}r^{2}f^{2}}{2\Omega^{2}}+q_{0}^{2}Qr\left(\frac{A}{\Omega^{2}}\right)\kappa f^{2}\label{fixpointeq3}\\&
\frac{dQ}{dr} = q_{0}^{2}r^{2}\kappa\frac{A}{\Omega^{2}}f^{2},\label{fixpointeq4}\\&
\frac{dA}{dr} = \frac{Q\kappa}{r^{2}},\label{fixpointeq5}
\end{align}
with the initial conditions:
\begin{equation}
\label{perturbed initial and gauge condition}
\begin{aligned}
&\varpi(r_{+}) = M,\quad Q(r_{+}) = e, \quad \kappa^{-1}(r_{+}) = 1-\frac{\bigl(-\frac{\Lambda}{3}\bigr)\alpha r_{+}}{2K_{+}}\epsilon^{2},\\ 
&\Omega^{2}\kappa^{-2}(r_{+}) = 1-\frac{2M}{r_{+}}+\frac{e^{2}}{r_{+}^{2}}+\bigl(-\frac{\Lambda}{3}\bigr)r_{+}^{2} = 0,\\&
\lim_{r\rightarrow r_{+}}\frac{A}{\Omega^{2}} = \frac{e}{2K_{+}r_{+}^{2}},\quad \frac{d\Omega^{2}}{dr}(r_{+}) = 2K_{+}\kappa(r_{+}).
\end{aligned}
\end{equation}
Let $(\kappa,\Omega^{2},\varpi,Q,A)$ be the solutions to the above equations \eqref{fixpointeq1}-\eqref{fixpointeq5} and $\phi$ be the solution of the equation:\begin{equation}
\frac{d}{dr}\left(r^{2}\Omega^{2}\kappa^{-1}\frac{d\phi}{dr}\right) = \left(\bigl(-\frac{\Lambda}{3}\bigr)\alpha-\frac{q_{0}^{2}A^{2}}{\Omega^{2}}\right)r^{2}\kappa\phi\label{fix point phi}
\end{equation}
with initial conditions\begin{equation}
\phi(r_{+}) = \epsilon,\quad \frac{d\phi}{dr}(r_{+}) = \frac{\bigl(-\frac{\Lambda}{3}\bigr)\alpha\epsilon}{2K_{+}}\kappa(r_{+}).
\end{equation}
The existence of $\phi$ for $\epsilon$ small enough has been established in Section \ref{sec:proof of linear theorem}. Thus we can define a map$$T:X_{\delta}\rightarrow C^{2}$$ by $T[f] = \phi$.

In the following, we will first show that for some parameters $(M,r_{+},\Lambda,\alpha,q_{0})$, $T$ is a map from $X_{\delta}$ to $X$. Then by choosing $\epsilon$ and $\delta$ appropriately, we further show that $T$ is a contraction.

Directly integrating \eqref{fixpointeq1}-\eqref{fixpointeq5} and using a standard bootstrap estimate as in Section \ref{sec:proof of linear theorem}, we can get the following estimates:
\begin{proposition}
For $\delta$ and $\epsilon$ small enough, and any $f\in X_{\delta}$, the solutions $(\kappa,\Omega^{2},\varpi,Q,A)$ to \eqref{fixpointeq1}-\eqref{fixpointeq5} with the initial conditions \eqref{perturbed initial and gauge condition} satisfy the following estimates:\begin{align}
1\leq \kappa(r_{+})\leq\kappa(r)\leq&\kappa(r_{+})+C\Vert f\Vert_{X}^{2},\label{basic estimate for fix point 1}\\
\vert\varpi-M\vert\leq&C\Vert f\Vert_{X}^{2},\\
\vert Q-e\vert\leq& C\Vert f\Vert_{X}^{2},\\
\vert A-A_{RN}\vert\leq& C\Vert f\Vert_{X}^{2},\\
\left\vert\frac{Ar^{2}}{\Omega^{2}\kappa^{-2}}-\frac{A_{RN}r^{2}}{\Omega_{RN}^{2}}\right\vert\leq& C\Vert f\Vert_{X}^{2}.\label{basic estimate for fix point2}
\end{align}
\end{proposition}
\begin{proof}
Analogously to the proofs of Lemma \ref{lem Kap}, Lemma \ref{Q}, Lemma \ref{varpi}, and Lemma \ref{estimate for ren A}.
\end{proof}

\subsection{Construction of $T:X_{\delta}\rightarrow X$}
To prove $T$ is a contraction, the first step is to show there exist parameters $(M,r_{+},\Lambda,\alpha,q_{0})$ such that $T$ is a map from $X_{\delta}$ to $X$. In view of the estimates \eqref{basic estimate for fix point 1}-\eqref{basic estimate for fix point2}, the equation \eqref{fix point phi} is analogous to the stationary spherically symmetric Klein--Gordon equations on Reissner--Nordström-AdS up to a $\delta^{2}$ perturbation, the existence of stationary solutions under Dirichlet boundary conditions has been established in our companion work \cite{weihaozheng}. We will state the following perturbed version of the main theorem in \cite{weihaozheng}, the proof of which can be found in the appendix.
\begin{theorem}
\label{thm: fix point stationary D}
Let $(r_{+},\Lambda,\alpha)$ be fixed parameters with $\Lambda<0$ and $-\frac{9}{4}<\alpha<0$, $f\in X_{\delta}$ be a fixed function and $(\kappa,\Omega^{2},\varpi,Q,A)$ be the solutions to \eqref{fixpointeq1}-\eqref{fixpointeq5}. Then for $\epsilon$ and $\delta$ small enough, we have the following two results:
\begin{enumerate}
  \item \textup{(Large charge case)} For any given $M_{e = 0}<M_{b}<M_{0}$, there exists a $q_{1} = q_{1}(M_{b},r_{+},\Lambda,\alpha,f,\delta,\epsilon)$ such that for any $\vert q_{0}\vert>q_{1}$, there exists an initial condition $M_{e = 0}<\varpi(r_{+}) = M_{c}^{f} = M_{c}^{f}(M_{b},r_{+},\Lambda,\alpha,q_{0},f,\delta,\epsilon)<M_{b}$, such that a bounded stationary solution to \eqref{fix point phi} under the Dirichlet boundary condition exists.
  \item \textup{(Fixed charge case)} If the fixed parameters $(r_{+},\Lambda,\alpha)$ further satisfy the conditions\begin{align}
  &-\frac{9}{4}<\alpha<\min\left\{-\frac{3}{2}+\frac{q_{0}^{2}}{2\bigl(-\frac{\Lambda}{3}\bigr)},0\right\},\\&
  \bigl(-\frac{\Lambda}{3}\bigr)r_{+}^{2}>R_{0},
  \end{align}
  where $R_{0}$ is the positive solution to the quadratic equation \eqref{quadratic}, then there exists an initial condition $\varpi(r_{+}) = M_{c}^{f} = M_{c}^{f}(r_{+},\Lambda,\alpha,f,\delta,\epsilon)\in (M_{e = 0},M_{0})$ such that a bounded stationary solution to \eqref{fix point phi} under the Dirichlet boundary condition exists.
\end{enumerate}
Moreover, letting $M_{c}$ be the black hole mass defined in Theorem \ref{linear hair theorem}, we have the following estimate:\begin{equation}
\label{Mc estimate}
\vert M_{c}^{f}-M_{c}\vert\leq C\Vert f\Vert_{X}^{2},
\end{equation}
where $C$ is a constant depending on $(M_{b},r_{+},\Lambda,\alpha,q_{0})$ in the large charge case and depending on $(r_{+},\Lambda,\alpha,q_{0})$ for the fixed charge case.
\end{theorem}
\begin{proof}
See Appendix \ref{appendix}.
\end{proof}
\begin{remark}
In the later discussion, the notation $C$ refers to a universal constant, i.e, a constant only depending on $(M_{b},r_{+},\Lambda,\alpha,q_{0})$ in the large charge case and depending on $(r_{+},\Lambda,\alpha,q_{0})$ for the fixed charge case.
\end{remark}
\subsection{Estimate for $\Vert T[f]\Vert_{X}$}
Next, we estimate $\Vert T[f]\Vert_{X}$ for the parameters $(M_{c}^{f},r_{+},\Lambda,\alpha,q_{0})$ in Theorem \ref{thm: fix point stationary D}.
\begin{proposition}
\label{estimate for Tf}
For $\epsilon, \delta$ small enough and parameters $(M_{c}^{f},r_{+},\Lambda,\alpha,q_{0})$ in Theorem \ref{thm: fix point stationary D}, we have the estimate\begin{equation}
\Vert \phi\Vert_{X}\leq C(\Vert f\Vert_{X}^{2}+1)\epsilon.
\end{equation}
\end{proposition}
\begin{proof}
We can write the equation \eqref{fix point phi} as\begin{equation}
\label{error form}
\begin{aligned}
&-\frac{d}{dr}\left(r^{2}\Omega_{RN,M_{c}^{f}}^{2}\frac{d\phi}{dr}\right)+\left(\bigl(-\frac{\Lambda}{3}\bigr)\alpha-\frac{q_{0}^{2}A_{RN}^{2}}{\Omega_{RN,M_{c}^{f}}^{2}}\right)r^{2}\phi \\=&\frac{d}{dr}\left(r^{2}\left(\Omega^{2}\kappa^{-2}-\Omega_{RN,M_{c}^{f}}^{2}\right)\frac{d\phi}{dr}\right)+\left(\frac{q_{0}^{2}A^{2}}{\Omega^{2}}-\frac{q_{0}^{2}A_{RN}^{2}}{\Omega_{RN,M_{c}^{f}}^{2}}\right)r^{2}\phi+r^{2}\Omega^{2}\kappa^{-2}\kappa^{-1}\frac{d\kappa}{dr}\frac{d\phi}{dr}\\=&:E_{f}[\phi].
\end{aligned}
\end{equation}
Let \begin{equation*}
\phi = C_{D}(r)u_{D}(r)+C_{N}(r)u_{N}(r),
\end{equation*}
where $u_{D}$ and $u_{N}$ are defined as in Theorem \ref{linear theory}. Let $\widetilde{C}_{D}$ and $\widetilde{C}_{N}$ be two constants such that \begin{equation*}
\lim_{r\rightarrow r_{+}}\widetilde{C}_{D}u_{D}(r)+\widetilde{C}_{N}u_{N}(r) = 1.
\end{equation*}
Since $\phi$ satisfies the Dirichlet boundary condition and $\phi(r_{+}) = \epsilon$, we have\begin{equation*}
\lim_{r\rightarrow\infty} C_{N}(r) = 0,\quad C_{D}(r_{+}) = \widetilde{C}_{D}\epsilon.
\end{equation*}
Then we can get equations for $C_{D}(r)$ and $C_{N}(r)$ as in Section \ref{sec:proof of linear theorem}\begin{align*}
&C_{D}^{\prime} = \frac{-u_{N}E_{f}[\phi]}{r^{2}\Omega_{RN,M_{c}^{f}}^{2}(u_{D}^{\prime}u_{N}-u_{N}^{\prime}u_{D})},\\&
C_{N}^{\prime} = \frac{u_{D}E_{f}[\phi]}{r^{2}\Omega_{RN,M_{c}^{f}}^{2}(u_{D}^{\prime}u_{N}-u_{D}u_{N}^{\prime})}.
\end{align*}
We estimate $C_{N}$ first. By the Dirichlet boundary condition, we have\begin{equation}
\begin{aligned}
\vert C_{N}(r)\vert  \leq &\int_{r}^{\infty}\frac{\vert u_{D}\vert}{\bar{r}^{2}\Omega_{RN,M_{c}^{f}}^{2}\vert u_{D}^{\prime}u_{N}-u_{N}^{\prime}u_{D}\vert}\left\vert\frac{d}{dr}\left(r^{2}\left(\Omega^{2}\kappa^{-2}-\Omega_{RN,M_{c}^{f}}^{2}\right)\frac{d\phi}{dr}\right)\right\vert\\&
+\frac{\vert u_{D}\vert}{\bar{r}^{2}\Omega_{RN,M_{c}^{f}}^{2}\vert u_{D}^{\prime}u_{N}-u_{D}u_{N}^{\prime}\vert}\left\vert\left(\frac{q_{0}^{2}A^{2}}{\Omega^{2}}-\frac{q_{0}^{2}A_{RN}^{2}}{\Omega_{RN,M_{c}^{f}}^{2}}\right)\bar{r}^{2}\phi\right\vert\\&+
\frac{\vert u_{D}\vert}{\bar{r}^{2}\Omega_{RN,M_{c}^{f}}^{2}\vert u_{D}^{\prime}u_{N}-u_{N}^{\prime}u_{D}\vert}\left\vert\kappa^{-1}\frac{d\kappa}{dr}\bar{r}^{2}\Omega^{2}\kappa^{-2}\frac{d\phi}{dr}\right\vert\mathrm{d}\bar{r}
\end{aligned}
\label{Cd equation}
\end{equation}
Keeping track of the constant dependence in Section \ref{sec:proof of linear theorem}, we have\begin{equation}
\left\vert r^{\frac{3}{2}-\Delta}\phi\right\vert+\left\vert r^{\frac{5}{2}-\Delta}\frac{d\phi}{dr}\right\vert+\left\vert r^{\frac{3}{2}-\Delta}\Omega_{RN,M_{c}^{f}}^{2}\frac{d^{2}\phi}{dr^{2}}\right\vert\leq C\epsilon.\label{basic estimate}
\end{equation}
We can estimate $\eqref{Cd equation}$ the same way as in Section $\ref{sec:proof of linear theorem}$. For the first term on the right hand side of $\eqref{Cd equation}$, we have\begin{equation}
\label{Cd estimate1}
\begin{aligned}
&\int_{r}^{\infty}\frac{\vert u_{D}\vert}{\bar{r}^{2}\Omega_{RN,M_{c}^{f}}^{2}\vert u_{D}^{\prime}u_{N}-u_{D}u_{N}^{\prime}\vert}\left\vert\frac{d}{dr}\left(\left(-2(\varpi-M_{c}^{f})\bar{r}+(Q^{2}-e^{2})\right)\frac{d\phi}{dr}\right)\right\vert\mathrm{d}\bar{r}\\\leq&\int_{r}^{\infty}\frac{\vert u_{D}\vert}{\bar{r}^{2}\Omega_{RN,M_{c}^{f}}^{2}\vert u_{D}^{\prime}u_{N}-u_{D}u_{N}^{\prime}\vert}\left(2\vert\varpi-M_{c}^{f}\vert+\vert Q^{2}-e^{2}\vert\right)\left\vert\frac{d^{2}\phi}{dr^{2}}\right\vert\\&
+\frac{\vert u_{D}\vert}{\bar{r}^{2}\Omega_{RN,M_{c}^{f}}^{2}\vert u_{D}^{\prime}u_{N}-u_{D}u_{N}^{\prime}\vert}\left(2\left\vert\frac{d\varpi}{dr}\right\vert+2\vert\varpi-M_{c}^{f}\vert+2\left\vert Q\frac{dQ}{dr}\right\vert\right)\left\vert\frac{d\phi}{dr}\right\vert\mathrm{d}\bar{r}\\\leq&
C\Vert f\Vert_{X}^{2}\epsilon r^{-4}.
\end{aligned}
\end{equation} 
For the second term on the right hand side of $\eqref{Cd equation}$, we have\begin{equation}
\label{Cd estimate2}
\int_{r}^{\infty}\frac{\vert u_{D}\vert}{\bar{r}^{2}\Omega_{RN,M_{c}^{f}}^{2}\vert u_{D}^{\prime}u_{N}-u_{D}u_{N}^{\prime}\vert}\left\vert\left(\frac{q_{0}^{2}A^{2}}{\Omega^{2}}-\frac{q_{0}^{2}A_{RN}^{2}}{\Omega_{RN,M_{c}^{f}}^{2}}\right)\bar{r}^{2}\phi\right\vert\leq C\epsilon\Vert f\Vert_{X}^{2}r^{-2}.
\end{equation}
For the third term on the right hand side of $\eqref{Cd equation}$, we have\begin{equation}
\label{Cd estimate3}
\int_{r}^{\infty}\frac{\vert u_{D}\vert}{\bar{r}^{2}\Omega_{RN,M_{c}^{f}}^{2}\vert u_{D}^{\prime}u_{N}-u_{D}u_{N}^{\prime}\vert}\left\vert\kappa^{-1}\frac{d\kappa}{dr}\bar{r}^{2}\Omega^{2}\kappa^{-2}\frac{d\phi}{dr}\right\vert\mathrm{d}\bar{r}\leq C \epsilon r^{-2-2\Delta}\Vert f\Vert_{X}^{2}
\end{equation}
Combining $\eqref{Cd estimate1}-\eqref{Cd estimate3}$, we have\begin{equation}
\vert C_{N}(r)\vert\leq C\epsilon\Vert f\Vert_{X}^{2}r^{-2}.
\end{equation}

Similarly for $C_{D}(r)$, we can get\begin{equation*}
\left\vert C_{D}(r)-\widetilde{C}_{D}\Vert f\Vert_{X}\right\vert\leq C\epsilon\Vert f\Vert_{X}^{2}r^{-2+2\Delta}.
\end{equation*}
Hence if $\Delta\leq1$, we have\begin{align*}
&\left\vert r^{\frac{3}{2}+\Delta}(C_{D}(r)u_{D}+C_{N}(r)u_{N})\right\vert\leq \left\vert r^{2\Delta}C_{D}(r)\right\vert \left\vert r^{\frac{3}{2}-\Delta}u_{D}\right\vert+\left\vert C_{N}(r)r^{\frac{3}{2}+\Delta}u_{N}\right\vert\leq C\left(\Vert f\Vert_{X}^{2}+1\right)\epsilon,\\&
\left\vert r^{\frac{5}{2}+\Delta}\left(C_{D}(r)u_{D}^{\prime}+C_{N}(r)u_{N}^{\prime}\right)\right\vert\leq \left\vert r^{2\Delta}C_{D}(r)\right\vert \left\vert r^{\frac{5}{2}-\Delta}u^{\prime}_{D}\right\vert+\left\vert C_{N}(r)r^{\frac{5}{2}+\Delta}u_{N}^{\prime}\right\vert\leq C\left(\Vert f\Vert_{X}^{2}+1\right)\epsilon,\\&
\left\vert r^{\frac{3}{2}+\Delta}\Omega_{RN,M_{c}^{f}}^{2}(C_{D}^{\prime}u_{D}^{\prime}+C_{N}^{\prime}u_{N}^{\prime})\right\vert\leq C(\Vert f\Vert_{X}^{2}+1)\epsilon.
\end{align*}
If $\Delta>1$, then similarly we have\begin{equation}
\left\vert r^{\frac{7}{2}-\Delta}\phi\right\vert+\left\vert r^{\frac{9}{2}-\Delta}\frac{d\phi}{dr}\right\vert+\left\vert r^{\frac{7}{2}-\Delta}\Omega_{RN,M_{c}^{f}}^{2}\frac{d^{2}\phi}{dr^{2}}\right\vert\leq C(\Vert f\Vert_{X}^{2}+1)\epsilon.\label{improved estimate}
\end{equation}
Then using \eqref{improved estimate} instead of \eqref{basic estimate} in the above estimate, we can conclude that\begin{equation}
\Vert \phi\Vert_{X}\leq C\epsilon(1+\Vert f\Vert^{2}_{X}).
\end{equation}
\end{proof}
Then by Proposition \ref{estimate for Tf}, for $\epsilon=\delta$ small enough, $T$ is a map from $X_{\delta}$ to $X_{\delta}$.

\subsection{Estimate for $\Vert T[f_{1}]-T[f_{2}]\Vert_{X}$} To prove $T$ is a contraction, we still need to estimate $\Vert T[f_{1}]-T[f_{2}]\Vert$. Since for different $f\in X_{\delta}$, we will get different $M$ in Theorem \ref{thm: fix point stationary D}. The key estimate in this section will be the estimate for $M$, the idea of which has already been addressed in the proof of \eqref{Mc estimate} in Appendix \ref{appendix}. To be precise, we have the following lemma:
\begin{lemma}
\label{lemma: estimate for the difference of M}
Let $M_{c}^{f}$ be defined as in Theorem \ref{thm: fix point stationary D}. Then for $\epsilon,\delta$ small enough and any $f_{1},f_{2}\in X_{\delta}$, the following estimate holds:\begin{equation}
\vert M_{c}^{f_{1}}-M_{c}^{f_{2}}\vert\leq C\Vert f_{1}-f_{2}\Vert\delta^{3}.
\end{equation}
\end{lemma}
The proof of Lemma \ref{lemma: estimate for the difference of M} lies in two parts\begin{enumerate}
\item Estimate $(\kappa_{1}-\kappa_{2},\Omega_{1}^{2}-\Omega_{2}^{2},\varpi_{1}-\varpi_{2},Q_{1}-Q_{2},A_{1}-A_{2})$ where $(\kappa_{i},\Omega_{i},\varpi_{i},Q_{i},A_{i})$ are the solutions to \eqref{fixpointeq1}-\eqref{fixpointeq5} with the same $f$ but different initial conditions for $\varpi$; see Section \ref{sec1}.
\item Estimate $(\kappa_{1}-\kappa_{2},\Omega_{1}^{2}-\Omega_{2}^{2},\varpi_{1}-\varpi_{2},Q_{1}-Q_{2},A_{1}-A_{2})$ where $(\kappa_{i},\Omega_{i},\varpi_{i},Q_{i},A_{i})$ are the solutions to \eqref{fixpointeq1}-\eqref{fixpointeq5} with the same initial conditions but different $f$; see Section \ref{sec2}.
\end{enumerate}
Then using the argument in Lemma \ref{pro: estimate for different M} in Appendix \ref{appendix}, we can conclude the proof of Lemma \ref{lemma: estimate for the difference of M}.

\subsubsection{Estimate for \eqref{fixpointeq1}-\eqref{fixpointeq5} with the same $f$ but different initial conditions}
\label{sec1}
Let $(\kappa_{i},\Omega_{i}^{2},\varpi_{i},Q_{i},A_{i})$ be the solutions to \eqref{fixpointeq1}-\eqref{fixpointeq5} with initial conditions:
\begin{equation}
\label{difference initial}
\kappa_{i}^{-1}(r_{+}) = 1-\frac{\bigl(-\frac{\Lambda}{3}\bigr)\alpha r_{+}}{2(K_{+})_{i}}\epsilon^{2},\quad A_{i}(r_{+})=0,\quad \Omega_{i}^{2}(r_{+}) = 0,\quad \varpi_{i}(r_{+}) = M_{c}^{f_{i}}.
\end{equation}
By subtracting equations \eqref{fixpointeq1}-\eqref{fixpointeq5}, we can get the following difference equations:\begin{align*}
\kappa_{1}^{-1}\frac{d\kappa_{1}}{dr}-\kappa_{2}^{-2}\frac{d\kappa_{2}}{dr} =& q_{0}^{2}rf^{2}\left(\kappa_{1}^{2}\left(\frac{A_{1}}{\Omega_{1}^{2}}\right)^{2}-\kappa_{2}^{2}\left(\frac{A_{2}}{\Omega_{2}^{2}}\right)^{2}\right),\\
\Omega_{1}^{2}\kappa_{1}^{-2}-\Omega_{2}^{2}\kappa_{2}^{-2}=& -\frac{2}{r}\left(\varpi_{1}-\varpi_{2}\right)+\frac{Q_{1}^{2}-Q_{2}^{2}}{r^{2}},\\
\frac{dQ_{1}-Q_{2}}{dr} =& q_{0}^{2}r^{2}f^{2}\left(\kappa_{1}\frac{A_{1}}{\Omega_{1}^{2}}-\kappa_{2}\frac{A_{2}}{\Omega_{2}^{2}}\right),\\
\frac{dA_{1}-A_{2}}{dr} =& \frac{Q_{1}\kappa_{1}}{r^{2}}-\frac{Q_{2}\kappa_{2}}{r^{2}},\\
\frac{d\varpi_{1}-\varpi_{2}}{dr}+r\left(\frac{df}{dr}\right)^{2}(\varpi_{1}-\varpi_{2}) =& \frac{1}{2}\left(\frac{df}{dr}\right)^{2}(Q_{1}^{2}-Q_{2}^{2})+\frac{1}{2}q_{0}^{2}r^{2}f^{2}\left(\frac{A_{1}^{2}}{\Omega_{1}^{2}}-\frac{A_{2}^{2}}{\Omega_{2}^{2}}\right)\\&+q_{0}^{2}rf^{2}\left(\frac{Q_{1}\kappa_{1}A_{1}}{\Omega_{1}^{2}}-\frac{Q_{2}\kappa_{2}A_{2}}{\Omega_{2}^{2}}\right).
\end{align*}
We can also compute the initial data for $\frac{A_{i}r^{2}}{\Omega_{i}^{2}}$:\begin{equation*}
\lim_{r\rightarrow r_{+}}\frac{A_{i}r^{2}}{\Omega_{i}^{2}} = \frac{e_{i}}{2(K_{+})_{i}}.
\end{equation*}
Assuming $M_{c}^{f_{2}}<M_{c}^{f_{1}}$, we have the following estimates at $r = r_{+}$:\begin{align*}
&\varpi_{1}(r_{+})-\varpi_{2}(r_{+}) = M_{c}^{f_{1}}-M_{c}^{f_{2}},\\
&\frac{C}{2}(M_{c}^{f_{1}}-M_{c}^{f_{2}})\epsilon^{2}<\kappa_{1}(r_{+})-\kappa_{2}(r_{+})<2C(M_{c}^{f_{1}}-M_{c}^{f_{2}})\epsilon^{2},\\&
 \frac{C}{2}(M_{c}^{f_{1}}-M_{c}^{f_{2}})<Q_{1}-Q_{2}\leq 2C(M_{c}^{f_{1}}-M_{c}^{f_{2}}),\\&
\frac{C}{2}(M_{c}^{f_{1}}-M_{c}^{f_{2}})<\Bigl(\frac{A_{1}r^{2}}{\Omega_{1}^{2}}-\frac{A_{2}r^{2}}{\Omega_{2}^{2}}\Bigr)< 2C(M_{c}^{f_{1}}-M_{c}^{f_{2}}).
\end{align*}
Then by a bootstrap argument as in Section \ref{sec:proof of linear theorem}, we have the following proposition:
\begin{proposition}
\label{different M}
Let $(\kappa_{i},\Omega_{i}^{2},\varpi_{i},Q_{i},A_{i})$ be the solutions to \eqref{fixpointeq1}-\eqref{fixpointeq5} with initial conditions \eqref{difference initial}. Then for $\epsilon,\delta$ small enough and any $f\in X_{\delta}$, the following estimates hold\begin{align}
&\frac{C}{2}(M_{c}^{f_{1}}-M_{c}^{f_{2}})\epsilon^{2}<\kappa_{1}(r)-\kappa_{2}(r)<2C(M_{c}^{f_{1}}-M_{c}^{f_{2}}){\epsilon^{2}},\\&
\frac{C}{2}(M_{c}^{f_{1}}-M_{c}^{f_{2}})<Q_{1}-Q_{2}<2C(M_{c}^{f_{1}}-M_{c}^{f_{2}}),\\&
\frac{1}{2}(M_{c}^{f_{1}}-M_{c}^{f_{2}})<\varpi_{1}-\varpi_{2}<2(M_{c}^{f_{1}}-M_{c}^{f_{2}}),\\&
\frac{1}{2}C(M_{c}^{f_{1}}-M_{c}^{f_{2}})<\frac{A_{1}r^{2}}{\Omega_{1}^{2}\kappa_{1}^{-2}}-\frac{A_{2}r^{2}}{\Omega_{2}\kappa_{2}^{-2}}<2C(M_{c}^{f_{1}}-M_{c}^{f_{2}}),\\&
0\leq A_{1}-A_{2}<C(M_{c}^{f_{1}}-M_{c}^{f_{2}}).
\end{align}
\end{proposition}

\subsubsection{Estimate for \eqref{fixpointeq1}-\eqref{fixpointeq5} with same initial conditions but different $f$}
\label{sec2}
Next, we consider the difference between $(\kappa_{i},\Omega_{i}^{2},\varpi_{i},Q_{i},A_{i})$ for a given $M$ and different functions $f_{1},f_{2}\in X_{\delta}$. By the same bootstrap argument, we have\begin{proposition}
\label{different f}
Let $(\kappa_{i},\Omega_{i}^{2},\varpi_{i},Q_{i},A_{i})$ be the solutions to \eqref{fixpointeq1}-\eqref{fixpointeq5} with $f = f_{i}\in X_{\delta}$ under the same initial conditions. Then we have the following estimates:
\begin{align}
&\vert\kappa_{1}-\kappa_{2}\vert\leq C\delta^{3}\Vert f_{1}-f_{2}\Vert_{X},\\&
\vert\varpi_{1}-\varpi_{2}\vert\leq C\delta^{3}\Vert f_{1}-f_{2}\Vert_{X},\\&
\vert Q_{1}-Q_{2}\Vert \leq C\delta^{3}\Vert f_{1}-f_{2}\Vert_{X},\\&
\vert A_{1}-A_{2}\vert\leq C\delta^{3}\Vert f_{1}-f_{2}\Vert_{X}.
\end{align}
\end{proposition}
Now we are ready to prove Lemma \ref{lemma: estimate for the difference of M}:
\begin{proof}
Let $\phi_{i}$ be the corresponding solutions to \eqref{fix point phi} for $M = M_{c}^{f_{i}}$. Then by the argument in Appendix \ref{appendix}, we have\begin{align*}
0\leq& \mathcal{L}_{f_{2},M_{c}^{f_{2}}}[\phi_{1}]-\mathcal{L}_{f_{1},M_{c}^{f_{1}}}[\phi_{1}]\\ = &
\mathcal{L}_{f_{2},M_{c}^{f_{2}}}[\phi_{1}]-\mathcal{L}_{f_{2},M_{c}^{f_{1}}}[\phi_{1}]+\mathcal{L}_{f_{2},M_{c}^{f_{1}}}[\phi_{1}]-\mathcal{L}_{f_{1},M_{c}^{f_{1}}}[\phi_{1}].
\end{align*}
Without loss of generality, assuming $M_{c}^{f_{1}}>M_{c}^{f_{2}}$, then by Proposition \ref{different M}, we have \begin{equation}
\mathcal{L}_{f_{1},M_{c}^{f_{1}}}[\phi_{1}]-\mathcal{L}_{f_{1},M_{c}^{f^{2}}}[\phi_{1}]\geq C(M_{c}^{f_{1}}-M_{c}^{f_{2}})\int_{r_{+}}^{\infty}r^{3}(r-r_{+})\left(\frac{d\phi_{1}}{dr}\right)^{2}\mathrm{d}r.
\end{equation}
By Proposition \ref{different f}, we have \begin{equation}
\left\vert\mathcal{L}_{f_{2},M_{c}^{f^{1}}}[\phi_{1}]-\mathcal{L}_{f_{1},M_{c}^{f_{1}}}[\phi_{1}]\right\vert\leq C\delta^{3}\Vert f_{1}-f_{2}\Vert_{X}\int_{r_{+}}^{\infty}r^{3}(r-r_{+})\left(\frac{d\phi_{1}}{dr}\right)^{2}+r^{2}\phi_{1}^{2}\mathrm{d}r.
\end{equation}
Since \begin{align*}
\frac{C}{2}\int_{r_{+}}^{\infty}r^{3}(r-r_{+})\left(\frac{d\phi_{1}}{dr}\right)^{2}\mathrm{d}r&\leq\int_{r_{+}}^{\infty}r^{2}\Omega_{1}^{2}\kappa_{1}^{-1}\left(\frac{d\phi_{1}}{dr}\right)^{2}\\&=
\int_{r_{+}}^{\infty}\left(-\bigl(-\frac{\Lambda}{3}\bigr)\alpha+\frac{q_{0}^{2}A_{1}^{2}}{\Omega_{1}^{2}}\right)\kappa_{1}r^{2}\phi_{1}^{2}\mathrm{d}r\\&\leq \int_{r_{+}}^{\infty} 2Cr^{2}\phi_{1}^{2}\mathrm{d}r.
\end{align*}
Hence we have $$0\leq M_{c}^{f_{1}}-M_{c}^{f_{2}}\leq C\Vert f_{1}-f_{2}\Vert_{X}\delta^{3}.$$
\end{proof}
Now, we are ready to finish the estimate for $\Vert \phi_{1}-\phi_{2}\Vert_{X}$. We can subtract the equations \eqref{fix point phi} for $\phi_{1}$ and $\phi_{2}$ and write it in the error form. Let $v = \phi_{1}-\phi_{2}$, then we have\begin{equation}
\label{differror}
\begin{aligned}
&-\frac{d}{dr}\left(r^{2}\Omega_{1}^{2}\kappa_{1}^{-1}\frac{dv}{dr}\right)+\left(\bigl(-\frac{\Lambda}{3}\bigr)\alpha-\frac{q_{0}^{2}A_{1}^{2}}{\Omega_{1}^{2}}\right)r^{2}v\\=& \frac{d}{dr}\left(r^{2}\left(\Omega_{1}^{2}\kappa_{1}^{-1}-\Omega_{2}^{2}\kappa_{2}^{-1}\right)\frac{d\phi_{2}}{dr}\right)-\left(\bigl(-\frac{\Lambda}{3}\bigr)(\kappa_{1}-\kappa_{2})-q_{0}^{2}\left(\frac{A_{1}^{2}\kappa_{1}}{\Omega_{1}^{2}}-\frac{A_{2}^{2}\kappa_{2}}{\Omega_{2}^{2}}\right)\right)r^{2}\phi_{2}.
\end{aligned}
\end{equation}
We view the left hand side of $\eqref{differror}$ as the main part and right hand side of $\eqref{differror}$ as error terms. 

Now we are ready to estimate $\Vert\phi_{1}-\phi_{2}\Vert_{X}$. 

\begin{proposition}
For $\epsilon$ and $\delta$ small enough, T is a contraction.
\label{contraction}
\end{proposition}
\begin{proof}
Let$$v = \phi_{1}-\phi_{2},$$ we have\begin{align*}
&-\frac{d}{dr}\left(r^{2}\Omega_{1}^{2}\kappa_{1}^{-1}\frac{dv}{dr}\right)+\left(\bigl(-\frac{\Lambda}{3}\bigr)\alpha-\frac{q_{0}^{2}A_{1}^{2}}{\Omega_{1}^{2}}\right)\kappa_{1}r^{2}v = G_{error}\\&:=\frac{d}{dr}\left(r^{2}\left(\Omega_{1}^{2}\kappa_{1}^{-1}-\Omega_{2}^{2}\kappa_{2}^{-1}\right)\frac{d\phi_{2}}{dr}\right)-\left(\bigl(-\frac{\Lambda}{3}\bigr)(\kappa_{1}-\kappa_{2})-q_{0}^{2}\left(\frac{A_{1}^{2}\kappa_{1}}{\Omega_{1}^{2}}-\frac{A_{2}^{2}\kappa_{2}}{\Omega_{2}^{2}}\right)\right)r^{2}\phi_{2}
\end{align*}
Let \begin{equation}
v = C_{D}(r)u_{D}+C_{N}(r)u_{N},
\end{equation}
where $\{u_{D},u_{N}\}$ are the local basis of the linear equation:\begin{equation*}
-\frac{d}{dr}\left(r^{2}\Omega_{1}^{2}\kappa_{1}^{-1}\frac{dv}{dr}\right)+\left(\bigl(-\frac{\Lambda}{3}\bigr)\alpha-\frac{q_{0}^{2}A_{1}^{2}}{\Omega_{1}^{2}}\right)\kappa_{1}r^{2}v = 0
\end{equation*}
defined in Theorem \ref{linear theory}.
We have
\begin{equation*}
\lim_{r\rightarrow\infty}C_{N}(r) = 0,\quad 0<\lim_{r\rightarrow\infty}C_{D}(r)<\infty.
\end{equation*}
We have\begin{align*}
\vert C_{N}(r)\vert \leq& \int_{r}^{\infty}\frac{\vert u_{D}\vert}{\vert u_{D}^{\prime}u_{N}-u_{D}u_{N}^{\prime}\vert}\frac{\vert G_{error}\vert}{\bar{r}^{2}\Omega_{1}^{2}\kappa_{1}^{-2}}\mathrm{d}\bar{r}\\\leq&
\int_{r}^{\infty}\frac{\vert u_{D}\vert}{\vert u_{D}^{\prime}u_{N}-u_{D}u^{\prime}_{N}\vert}\frac{1}{\Omega_{1}^{2}\kappa_{1}^{-2}}\left\vert\kappa_{1}\left(\kappa_{1}^{-1}\frac{d\kappa_{1}}{dr}\Omega_{1}^{2}\kappa_{1}^{-2}-\kappa_{2}^{-1}\frac{d\kappa_{2}}{dr}\Omega_{2}^{2}\kappa_{2}^{-2}\right)\frac{d\phi_{2}}{dr}\right\vert\\&
+
\frac{\vert u_{D}\vert}{\vert u_{D}^{\prime}u_{N}-u_{D}u_{N}^{\prime}\vert}\frac{1}{\Omega_{1}^{2}\kappa_{1}^{-2}}\left\vert\kappa_{1}\left(\Omega_{1}^{2}\kappa_{1}^{-2}-\Omega_{2}^{2}\kappa_{2}^{-2}\right)\frac{d^{2}\phi_{2}}{dr^{2}}\right\vert\\&+
\frac{2\vert u_{D}\vert}{\vert u_{D}^{\prime}u_{N}-u_{D}u_{N}^{\prime}\vert}\frac{1}{\bar{r}\Omega_{1}^{2}\kappa_{1}^{-2}}\left\vert\kappa_{1}\left(\Omega_{1}^{2}\kappa_{1}^{-2}-\Omega_{2}^{2}\kappa_{2}^{-2}\right)\frac{d\phi_{2}}{dr}\right\vert\\&+
\frac{\vert u_{D}\vert}{\vert u_{D}^{\prime}u_{N}-u_{D}u_{N}^{\prime}\vert}\frac{1}{\Omega_{1}^{2}\kappa_{1}^{-2}}\left\vert\frac{d}{dr}\left(\Omega_{1}^{2}\kappa_{1}^{-2}-\Omega_{2}^{2}\kappa_{2}^{-2}\right)\frac{d\phi_{2}}{dr}\right\vert\\&+\frac{\vert u_{D}\vert}{\vert u_{D}^{\prime}u_{N}-u_{D}u_{N}^{\prime}\vert}\frac{\bigl(-\frac{\Lambda}{3}\bigr)\vert\kappa_{1}-\kappa_{2}\vert \vert\phi_{2}\vert^{2}}{\Omega_{1}^{2}\kappa_{1}^{-2}}\\&+\frac{\vert u_{D}\vert}{\vert u_{D}^{\prime}u_{N}-u_{D}u_{N}^{\prime}\vert}\frac{q_{0}^{2}\left\vert\frac{A_{1}^{2}\kappa_{1}}{\Omega_{1}^{2}}-\frac{A_{2}^{2}\kappa_{2}}{\Omega_{2}^{2}}\right\vert\vert\phi_{2}\vert}{\Omega_{1}^{2}\kappa_{1}^{-2}}\mathrm{d}\bar{r}
\end{align*}

We can estimate all these terms on the right hand side of the above inequality similarly to Section \ref{sec:proof of linear theorem}, using Lemma \ref{lemma: estimate for the difference of M} and Proposition \ref{basic estimate}. We will do the estimate for the second term to illustrate the method.\begin{align*}
&\int_{r}^{\infty}\frac{\vert u_{D}\vert}{\vert u_{D}^{\prime}u_{N}-u_{D}u_{N}^{\prime}\vert}\frac{1}{\Omega_{1}^{2}\kappa_{1}^{-2}}\left\vert\kappa_{1}\left(\Omega_{1}^{2}\kappa_{1}^{-2}-\Omega_{2}^{2}\kappa_{2}^{-2}\right)\frac{d^{2}\phi_{2}}{dr^{2}}\right\vert\mathrm{d}\bar{r}\\
\leq& C\int_{r}^{\infty}\frac{\vert u_{D}\vert}{\vert u_{D}^{\prime}u_{N}-u_{D}u_{N}^{\prime}\vert}\frac{1}{\Omega_{1}^{2}\kappa_{1}^{-2}}\left\vert\kappa_{1}\left(\frac{-2(\varpi_{1}-\varpi_{2})}{r}+\frac{Q_{1}^{2}-Q_{2}^{2}}{r^{2}}\right)\frac{d^{2}\phi_{2}}{dr^{2}}\right\vert\mathrm{d}\bar{r}\\\leq&
C\int_{r_{+}}^{\infty}\frac{\vert u_{D}\vert}{\vert u_{D}^{\prime}u_{N}-u_{D}u_{N}^{\prime}\vert}\frac{1}{\Omega_{1}^{2}\kappa_{1}^{-2}}\left(\frac{2\vert\varpi_{1}-M_{c}^{f_{1}}\vert+2\vert\varpi_{2}-M_{c}^{f_{2}}\vert}{r}\right)\frac{1}{\Omega_{RN,M_{c}^{f_{2}}}^{2}}\left\vert\frac{d^{2}\phi_{2}}{dr^{2}}\right\vert\mathrm{d}r\\&+
C\int_{r_{+}}^{\infty}\frac{\vert u_{D}\vert}{\vert u_{D}^{\prime}u_{N}-u_{D}u_{N}^{\prime}\vert}\frac{1}{\Omega_{1}^{2}\kappa_{1}^{-2}}\left(\frac{\vert Q_{1}^{2}-e_{1}^{2}\vert+\vert Q_{2}^{2}-e_{2}^{2}\vert}{r^{2}}\right)\frac{1}{\Omega_{RN,M_{c}^{f_{2}}}^{2}}\left\vert\frac{d^{2}\phi_{2}}{dr^{2}}\right\vert\mathrm{d}r\\&+C\int_{r_{+}}^{\infty}\frac{\vert u_{D}\vert}{\vert u_{D}^{\prime}u_{N}-u_{D}u_{N}^{\prime}\vert}\frac{1}{\Omega_{1}^{2}\kappa_{1}^{-2}}\left\vert\frac{2M_{c}^{f_{1}}-2M_{c}^{f_{2}}}{r}-\frac{Q_{1}^{2}-Q_{2}^{2}}{r^{2}}\right\vert\frac{1}{\Omega_{RN,M_{c}^{f_{2}}}^{2}}\left\vert\frac{d^{2}\phi_{2}}{dr^{2}}\right\vert\mathrm{d}r
\\\leq&
C\delta^{2}\Vert f_{1}-f_{2}\Vert_{X}r^{-2\Delta}.
\end{align*}
All other terms follow similarly. Thus we have\begin{equation}
\vert C_{N}(r)\vert\leq C\delta^{2}\Vert f_{1}-f_{2}\Vert_{X}r^{-2\Delta}.
\end{equation}
Similarly, we can estimate $C_{N}$:\begin{equation}
\vert C_{D}(r)\vert\leq C\delta^{2}\Vert f_{1}-f_{2}\Vert_{X}.
\end{equation}
Hence we have \begin{equation}
\Vert \phi_{1}-\phi_{2}\Vert_{X}\leq C\delta^{2}\Vert f_{1}-f_{2}\Vert_{X}.
\end{equation}
Then by choosing $\delta$ and $\epsilon$ small, we have $T$ is a contraction.
\end{proof}
\section{Proof of Theorem \ref{main theorem} under Neumann boundary conditions}
The proof of Theorem \ref{main theorem} under Neumann boundary conditions is essentially the same as the proof for Dirichlet boundary condition case, since Theorem \ref{linear hair theorem} also holds for Neumann boundary conditions. To avoid tedious computation, we only state the function space and the twisted equation for $\phi$. 

Define the norm $\Vert\cdot\Vert_{Y}$
\begin{equation*}
\Vert f\Vert_{Y} = \Vert r^{\frac{3}{2}-\Delta}f\Vert_{\infty}+\left\Vert r^{\frac{5}{2}-\Delta}\frac{df}{dr}\right\Vert_{\infty}+\left\Vert r^{\frac{3}{2}-\Delta}\Omega_{RN}^{2}\frac{d^{2}f}{dr^{2}}\right\Vert_{\infty}.
\end{equation*}
Let space $Y$ be the completion of all $f\in C^{2}(r_{+},\infty)$ with $\Vert f\Vert_{Y}<\infty$. We can rewrite the equation for $\phi$ \eqref{fix point phi} as \begin{equation}
\begin{aligned}
r^{\frac{3}{2}-\Delta}\frac{d}{dr}\left(r^{2\Delta-1}\Omega^{2}\kappa^{-1}\frac{dr^{\frac{3}{2}-\Delta}\phi}{dr}\right) =&-r^{\frac{3}{2}-\Delta}\frac{d}{dr}\left(\left(-\frac{3}{2}+\Delta\right)r^{-\frac{1}{2}+\Delta}\Omega^{2}\kappa^{-1}\right)\phi\\&+\left(\bigl(-\frac{\Lambda}{3}\bigr)\alpha-\frac{q_{0}^{2}A^{2}}{\Omega^{2}}\right)r^{2}\kappa\phi,
\end{aligned}
\end{equation}
which is called the twisted form. Then the proof of Theorem \ref{main theorem} follows under Neumann boundary condition follows line by line from Section \ref{proof:Dirichlet}.

\paragraph{Conflict of Interest Statement} The author declares that there is no conflict of interest.
\paragraph{Data Availability Statement} The author declares that there is no associated data.

\section{Appendix}
\label{appendix}
In the Appendix, we sketch the proof of Theorem \ref{thm: fix point stationary D}, which is very similar to the proof of Theorem \ref{linear hair theorem} established in our companion work \cite{weihaozheng}.

First, we can write down the energy functional of \eqref{fix point phi}:\begin{align*}
\mathcal{L}_{f,M}[g] =&\int_{r_{+}}^{\infty}r^{2}\Omega^{2}\kappa^{-1}\left(\frac{dg}{dr}\right)^{2}+\left(\bigl(-\frac{\Lambda}{3}\bigr)\alpha-\frac{q_{0}^{2}A^{2}}{\Omega^{2}}\right)r^{2}\kappa g^{2}\mathrm{d}r
\\
=&\int_{r_{+}}^{\infty}r^{2}\Omega_{RN,M}^{2}\kappa(r_{+})\left(\frac{dg}{dr}\right)^{2}+\left(\bigl(-\frac{\Lambda}{3}\bigr)\alpha-\frac{q_{0}^{2}A_{RN}^{2}}{\Omega_{RN,M}^{2}}\right)r^{2}\kappa(r_{+})g^{2}+F_{error}(r)\mathrm{d}r\\=&:L_{M}[g]+E[g],\\
F_{error}(r) =& (\kappa-\kappa(r_{+}))\left(r^{2}\Omega^{2}\kappa^{-2}\left(\frac{dg}{dr}\right)^{2}+\left(\bigl(-\frac{\Lambda}{3}\bigr)\alpha-\frac{q_{0}^{2}A^{2}}{\Omega^{2}\kappa^{-2}}\right)r^{2}g^{2}\right)\\&+r^{2}\left(\Omega^{2}\kappa^{-2}-\Omega_{RN,M}^{2}\right)\kappa(r_{+})\left(\frac{dg}{dr}\right)^{2}+q_{0}^{2}\left(\frac{A_{RN}^{2}}{\Omega_{RN,M}^{2}}-\frac{A^{2}}{\Omega^{2}}\right)\kappa(r_{+})r^{2}g^{2}.
\end{align*}
We say the energy functional $\mathcal{L}_{f,M}$ ($L_{M}$ respectively) admits a negative energy bound state if there exists $g\in C_{c}(r_{+},\infty)$ such that $\mathcal{L}_{f,M}[g]<0$ ($L_{M}[g]<0$ respectively). In view of the estimates \eqref{basic estimate for fix point 1}-\eqref{basic estimate for fix point2}, $\mathcal{L}_{f,M}$ can be treated as a $(\epsilon,\delta)$-perturbation of the energy functional $L_{M}$ for the Klein--Gordon equation \begin{equation}
-\frac{d}{dr}\left(r^{2}\Omega_{RN}^{2}\frac{d\phi}{dr}\right)+\left(\bigl(-\frac{\Lambda}{3}\bigr)\alpha-\frac{q_{0}^{2}A_{RN}^{2}}{\Omega_{RN}^{2}}\right)r^{2}\phi = 0.
\label{equ: unperturbed Klein Gordon}
\end{equation}

 We have the following two lemmas regarding the negative energy bound state of $L_{M}$:
\begin{lemma}\textup{\cite{weihaozheng}}
\label{charged negative energy bound state}
For any fixed sub-extremal parameters $(M_{b},r_{+},\Lambda,\alpha)$ satisfying the bound $-\frac{9}{4}<\alpha<0$, there exists a $q_{1}(M_{b},r_{+},\Lambda,\alpha)>0$ such that for any $\vert q_{0}\vert>q_{1}$, there exists a negative energy bound state for the functional $L_{M_{b}}[f]$.
\end{lemma}
\begin{lemma}\textup{\cite{weihaozheng}}
\label{nonempty}
For each fixed parameters$ (r_{+},\Lambda,\alpha,q_{0})$ satisfying\begin{align}
&-\frac{9}{4}<\alpha<\min\{0,-\frac{3}{2}+\frac{q_{0}^{2}}{2\bigl(-\frac{\Lambda}{3}\bigr)}\},\\&
\bigl(-\frac{\Lambda}{3}\bigr)r_{+}^{2}>R_{0},
\end{align}
where $R_{0}$ is the positive root of \eqref{quadratic}. Then there exists a negative energy bound state for the energy functional $L_{M_{0}}[f]$.
\end{lemma}
By the estimates \eqref{basic estimate for fix point 1}-\eqref{basic estimate for fix point2}, we conclude that if $\epsilon$ and $\delta$ are small, $\mathcal{L}_{f,M}[g]$ admits a negative energy bound state for $M = M_{b}$ if one considers the large charge case and $M = M_{0}$ if one considers the fixed charge case.

We consider the function class \begin{equation*}
\mathcal{F}: = \{f\in C_{c}^{\infty}(r_{+},\infty),\ \Vert f\Vert_{L^{2} = 1}\}.
\end{equation*}
Since the error term $E[g]$ is a $(\epsilon,\delta)$ perturbation, we can still rely on the sign property of the potential term to estiblish the following proposition, analogously to Proposition 6.13 in \cite{weihaozheng}.
\begin{proposition}
If $\mathcal{L}_{f,M}$ admits a negative energy bound state, there exist $\lambda_{M}>0$ and a non-zero solution $\phi_{f,M}$ of the equation\begin{equation}
-\frac{d}{dr}\left(r^{2}\Omega^{2}\kappa^{-1}\frac{d\phi}{dr}\right)+\left(\bigl(-\frac{\Lambda}{3}\bigr)\alpha-\frac{q_{0}^{2}A^{2}}{\Omega^{2}}\right)\kappa r^{2}\phi = -\lambda_{M}\phi
\end{equation}
satisfying the Dirichlet boundary condition and $\phi_{M}(r_{+}) = 1$.
\end{proposition}
Let $$\mathcal{A}_{s} = \{M_{e = 0}<M<s,\ \mathcal{L}_{f,M^{\prime}} \text{ has a negative energy bound state for } \forall M^{\prime}\in(M,s)\},$$ where $s = M_{0}$ for the general fixed charge case and $M = M_{b}$ for the large charge case.

Analogously to the proof of Theorem 4.4 in \cite{weihaozheng}, we can show that $\lambda_{M^{f}_{c}} = 0$ for $M^{f}_{c} = \inf\mathcal{A}_{s}$.

Let $M_{c}$ be the value of $M$ where a regular solution $\phi_{M_{c}}$ to \eqref{equ: unperturbed Klein Gordon} under the Dirichlet boundary condition exists. It remains to estimate $M_{c}^{f}$ in terms of $M_{c}$. We have the following proposition:
\begin{proposition}
\label{pro: estimate for different M}
For $f\in X_{\delta}$, the following estimate holds:\begin{equation}
\vert M_{c}^{f}-M_{c}\vert\leq C\Vert f\Vert_{X}^{2},
\end{equation}
where $C$ is a constant only depending on $(r_{+},\Lambda,\alpha,q_{0})$.
\end{proposition}
To prove Proposition \ref{pro: estimate for different M}, we need the following lemma:
\begin{lemma}
If $\mathcal{L}_{f,M}$ is always positive for $g\in C_{c}^{\infty}(r_{+},\infty)$, then $\mathcal{L}_{f,M}$ is always positve for any smooth function $g\in C_{c}^{\infty}[r_{+},\infty)$.
\end{lemma}
\begin{proof}
We only need to show that for any $\eta>0$ small, there exists $g_{\eta}\in C_{c}[r_{+},\infty)$ with $g(r_{+}) = 1$ such that $\vert L_{M}[g_{\eta}]\vert<\eta$. Let $g_{\eta} =1-\frac{1}{\epsilon^{\delta}}(r-r_{+})^{\delta}$ on $[r_{+},r_{+}+\epsilon]$ and $0$ for $r>r_{+}+\epsilon$. Then by direct computation, for $\epsilon$ and $\delta$ small enough, we have $vert L_{M}[g_{\eta^{\prime}}]\vert<\eta$.
\end{proof}
Now we are ready to prove Proposition \ref{pro: estimate for different M}
\begin{proof}
By the construction of $M_{c}^{f}$ and $\phi_{f,M_{c}^{f}}$, we have\begin{align*}
0&\leq\mathcal{L}_{f,M_{c}^{f}}[\phi_{M_{c}}]-L_{M_{c}}[\phi_{M_{c}}]\\&
=L_{M_{c}^{f}}[\phi_{M_{c}}]+E[\phi_{M_{c}}]-L_{M_{c}}[\phi_{M_{c}}]\\&=
\kappa(r_{+})\int_{r_{+}}^{\infty}r^{2}\left(\Omega_{RN,M_{c}^{f}}^{2}-\Omega_{RN,M_{c}}^{2}\right)\left(\frac{d\phi_{M_{c}}}{dr}\right)^{2}+q_{0}^{2}\left(\frac{A_{RN,M_{c}}^{2}}{\Omega_{RN,M_{c}}^{2}}-\frac{A_{RN,M_{c}^{f}}^{2}}{\Omega_{RN,M_{c}^{f}}^{2}}\right)r^{2}\phi_{M_{c}}^{2}\mathrm{d}r+E[\phi_{M_{c}}].
\end{align*}
If $M_{c}^{f}>M_{c}$, then we have \begin{align*}
2\kappa(r_{+})(M_{c}^{f}-M_{c})\int_{r_{+}}^{\infty}\left(r-r_{+}\right)\left(\frac{d\phi_{M_{c}}}{dr}\right)^{2}\leq C\Vert f\Vert_{X}^{2}\int_{r_{+}}^{\infty}r^{4}\left(\frac{d\phi_{M_{c}}}{dr}\right)^{2}+r^{2}\phi_{M_{c}}^{2}\mathrm{d}r.
\end{align*}
Similarly, if $M_{c}^{f}<M_{c}$, we have\begin{align*}
2\kappa(r_{+})(M_{c}-M_{c}^{f})\int_{r_{+}}^{\infty}\left(r-r_{+}\right)\left(\frac{d\phi_{M_{c}}}{dr}\right)^{2}\geq -C\Vert f\Vert_{X}^{2}\int_{r_{+}}^{\infty}r^{4}\left(\frac{d\phi_{M_{c}}}{dr}\right)^{2}+r^{2}\phi_{M_{c}}^{2}\mathrm{d}r.
\end{align*}
Since $\phi_{M_{c}}$ is a fixed function determined by $(r_{+},\Lambda,\alpha,q_{0})$, we have \begin{equation}
\vert M_{c}^{f}-M_{c}\vert\leq C\Vert f\Vert_{X}^{2}.
\end{equation}
\end{proof}

\end{document}